\documentclass[iicol]{sn-jnl}

\usepackage[T1]{fontenc}
\usepackage{amsmath,amssymb,amsfonts,amsthm}
\usepackage{booktabs}
\usepackage{multirow}
\usepackage{graphicx}
\usepackage{xcolor}
\usepackage{algorithm}
\usepackage{algpseudocode}
\usepackage{tikz}
\usepackage{pgfplots}
\usepackage{url}
\usepackage{hyperref}
\usepackage{thmtools}
\usepackage{pifont}
\usepackage{tabularx}

\newcommand{\method}{RareSense}
\newcommand{\methodfixed}{RareSense-Fixed}
\newcommand{\X}{\mathcal{X}}
\newcommand{\T}{\mathcal{T}}
\newcommand{\A}{\mathcal{A}}

\newcommand{\Rset}{\mathcal{R}}
\newcommand{\supp}{\operatorname{supp}}
\newcommand{\conf}{\operatorname{conf}}
\newcommand{\lift}{\operatorname{lift}}
\newcommand{\TopK}{\operatorname{TopK}}

\newcommand{\argmax}{\operatorname*{arg\,max}}

\theoremstyle{thmstyleone}%
\newtheorem{theorem}{Theorem}
\newtheorem{proposition}[theorem]{Proposition}%

\theoremstyle{thmstyletwo}%

\theoremstyle{thmstylethree}%

\begin{document}

\title[RareSense: Rarity-Aware Similarity for Anomaly Retrieval]{RareSense: Rarity-Aware Similarity Search for Anomaly Retrieval in Transactional Data}


\author*{\fnm{Sidahmed} \sur{Benabderrahmane}}\email{sidahmed.benabderrahmane@nyu.edu}

\author{\fnm{Talal} \sur{Rahwan}}\email{talal.rahwan@nyu.edu}

\affil{\orgaddress{New York University NYUAD,} \orgdiv{Department of Computer Science, Division of Science} }

\abstract{
Similarity search over sparse set-valued data is often dominated by frequent
background attributes because classical measures such as Jaccard, cosine, and Hamming compare objects through atomic overlap. IDF (Inverse document frequency) weighting
partially reduces this effect but remains atom-wise and cannot explicitly
represent informative higher-order co-occurrences. We introduce \method{}, a
rarity-aware similarity framework for sparse transactional anomaly data.
\method{} mines minimal rare itemsets as intermediate structures, derives
reliable rare association rules, maps objects into sparse rare-rule profiles,
and compares them using weighted Jaccard similarity. Rule weights combine
inverse support, confidence, lift, structural complexity, and stability, so
that neighborhoods are determined by shared rare evidence rather than uniform
feature overlap. We show that IDF-weighted Jaccard is a restricted singleton
case of \method{}, and that the induced distance is a pseudometric on the
original objects and a metric over equivalence classes defined by identical
rule profiles. Experiments across four benchmark families spanning cybersecurity and
general categorical domains show that \method{} attains the highest
observed macro-average query-conditioned retrieval performance among
the evaluated similarity measures. The statistical analysis indicates
significant overall differences, with corrected paired comparisons
favoring \method{} over the atomic baselines. The gains remain
workload-dependent and are strongest when anomalies share repeatable
rare higher-order structure. For global anomaly ranking, \method{}
achieves the highest observed macro-average performance while remaining
statistically comparable to several strong dedicated detectors.
}

\keywords{Similarity search, anomaly
retrieval, rare patterns, pattern-based
similarity, explainability}



\maketitle

\section{Introduction}
\label{sec:intro}

Similarity search is a fundamental primitive in data mining, information
retrieval, and anomaly analysis~\cite{zezula2006similarity,chavez2001searching,sakar25}.
Given a query object, a system returns objects that are close according to a
similarity model \cite{Xu25,Enen25}.  In anomaly-analysis workflows this operation appears in two
related forms \cite{Cheng25}.  First, suspicious objects may be globally ranked for inspection.
Second, once an analyst confirms an anomaly, that object can be used as a query
to retrieve related suspicious cases.  The second setting is especially
important operationally: analysts rarely inspect an entire ranking, so the
quality of the first few retrieved objects matters more than global separation
alone.

This paper studies sparse \emph{transactional} anomaly data.  Each object is
represented as a finite set of atoms, such as active binary features, one-hot
attribute--value indicators, discretized numerical bins, log templates, event
types, or session-level indicators \cite{Qiao25,Huang25,ammar2025foundation,haider2025survey}. Such representations
arise naturally in network-traffic and provenance-based security
analysis \cite{samariya2023comprehensive,qiao2025deep}.

Classical set and vector similarities compare objects through atomic overlap.
Jaccard~\cite{jaccard1912distribution}, Dice~\cite{dice1945measures}, Hamming \cite{cha2007comprehensive},
cosine \cite{santini2002similarity}, and IDF-weighted \cite{lan2022research} variants are natural choices.  IDF reduces the
influence of frequent atoms, but it remains atom-wise: the importance of an atom
does not depend on the other atoms with which it co-occurs.  This is limiting
when anomaly evidence is combinatorial.  For example, \texttt{login\_failure},
\texttt{outbound}, and \texttt{unusual\_port} may each be individually common
while their conjunction is rare and operationally meaningful.

We propose \method{}, a rarity-aware similarity model that moves the comparison
from the raw transactional space to a \emph{rare-rule evidence space}.  Rare
itemsets are used as an intermediate mining substrate.  Reliable association
rules derived from them define the final coordinates.  Each object is represented
by the retained rules whose complete evidence sets occur in the transaction, and
two objects are compared by the weighted overlap of those rule profiles.

The central distinction is between \emph{rare scoring} and \emph{rare
similarity}.  A pointwise rarity score asks:
\begin{quote}
Which objects contain strong rare evidence?
\end{quote}
\method{} asks:
\begin{quote}
Which objects share rare evidence with this query?
\end{quote}
The latter is a similarity-search question.  It can retrieve objects that are
related through the same rare conjunction even when their atomic overlap is
dominated by common background features.

A key methodological choice follows from this objective.  We treat
\textbf{nDCG@10 (Normalized Discounted Cumulative Gain) as the primary evaluation criterion} \cite{jarvelin2008discounted}.  AUROC is useful for
measuring global anomaly--normal separability, but it can remain high even when
all relevant anomalies occur too deep in the ranking to be inspected.  For
example, if every anomaly is ranked above almost every normal object but the
first anomaly appears only at rank 100, AUROC can be excellent while
nDCG@10 is zero.  For anomaly search, the latter reflects the operational
failure more directly.

The contributions are:

\begin{itemize}
    \item We introduce a rare-rule representation for similarity search over
    sparse transactional data.  Minimal rare itemsets are intermediate mining
    structures; reliable rare rules form the final similarity coordinates.

    \item We define a weighted-Jaccard similarity over rare-rule profiles and
    establish its metric interpretation on the profile space.  We also give a
    constructive result showing how higher-order coordinates can resolve
    atom-level similarity ties.

    \item We make pairwise explanation exact rather than post hoc: each shared
    rule has an additive contribution to the final similarity, and the top
    contributions form a faithful symbolic explanation of the retrieved
    neighbor.

    \item We define a reproducible adaptive configuration that uses an unlabeled
upper-tail rarity proxy to set the support ceiling and maximum pattern
length per workload. 

    \item We evaluate query-by-anomaly retrieval on 27 workloads from four
    benchmark families against four like-for-like query-conditioned atomic
    similarities---Jaccard, IDF-Jaccard, cosine, and TF--IDF cosine---and
    report scalar anomaly detectors only as secondary operational references.
    Global anomaly ranking is evaluated separately with AUROC.

    \item We identify a boundary condition of rarity-aware similarity: when a
    relevance label is already encoded by common atomic features, or when
    anomaly profiles share too little rare evidence, classical atomic
    similarities can be stronger. This limitation motivates hybrid
    atomic--rare similarity as a natural extension rather than obscuring a
    negative result.
\end{itemize}

\section{Related Work}
\label{sec:related}

RareSense lies at the intersection of similarity search, rare-pattern mining,
anomaly detection, and interpretable retrieval. These areas have largely
developed independently. Similarity-search methods typically operate on the
original feature representation; rare-pattern methods primarily aim to discover
unusual structures or convert them into scalar anomaly scores; and conventional
anomaly detectors rank objects according to query-independent abnormality.
Our work connects these perspectives by using statistically qualified rare
patterns to \emph{define the similarity representation itself}. Consequently,
the same symbolic representation supports query-conditioned retrieval, global
anomaly ranking, and explanations based on the rare evidence shared by a query
and its retrieved candidates.

\subsection{Set similarity and similarity search}

Set-based similarity is a fundamental primitive for comparing binary,
categorical, and transactional objects \cite{amer2020set}. Classical coefficients such as
Jaccard~\cite{jaccard1912distribution} and Dice~\cite{dice1945measures}
measure the degree of overlap between two sets, while cosine and
Hamming-derived measures provide related comparisons for binary vector
representations \cite{norouzi2012hamming}. These methods are simple, interpretable, and computationally
attractive, but all observed attributes are treated as the basic coordinates of
the comparison. Consequently, similarity is determined primarily by overlap in
individual atoms rather than by whether objects share higher-order behavioral
structures.

Term-weighting schemes provide a first mechanism for distinguishing common from
informative attributes. In particular, IDF-style weighting
~\cite{salton1988term} reduces the contribution of frequently occurring
features. Weighted Jaccard \cite{li2021rejection} and related Tanimoto formulations generalize
set overlap to non-uniform, non-negative feature weights
~\cite{marczewski1958certain,levandowsky1971distance,willett1998chemical}.
Efficient approximation techniques such as weighted MinHash further make such
weighted similarities practical at scale~\cite{ioffe2010improved}. These
approaches, however, still assume that the coordinates being compared are given
in advance. Weighting changes the importance of a feature, but not the
semantics of the feature space itself.

A complementary line of research addresses the efficient execution of
similarity queries. In metric spaces, the triangle inequality enables pruning
through index structures such as M-trees~\cite{ciaccia1997mtree} and
VP-trees~\cite{yianilos1993vptree}. More general similarity spaces often rely
on sequential scans, filtering, hashing, or approximate search \cite{zezula2006similarity,chavez2001searching}. 

Locality-sensitive hashing
and MinHash are classical examples of approximation mechanisms for
nearest-neighbor and Jaccard-like retrieval
~\cite{broder1997resemblance,indyk1998ann}. These methods focus primarily on
\emph{how} to search efficiently once a similarity representation has been
specified.

RareSense addresses a different and complementary question: \emph{what should
the coordinates of the similarity space represent when similarity is intended
to retrieve anomalous objects?} Rather than comparing objects directly in the
original atomic feature space, we construct a symbolic coordinate system whose
dimensions correspond to retained rare-rule evidence. Two objects are therefore
similar when they share statistically informative rare structures, not merely
when they share many individual attributes. This distinction is central:
RareSense is not simply another weighting of Jaccard over the original
features; it first transforms the representation through rare-pattern mining
and then performs weighted similarity in the resulting pattern space.

The resulting representation also preserves a useful connection to classical
similarity-search machinery. Once the rare-pattern dictionary is fixed, each
object is mapped to a sparse binary activation profile, and weighted Jaccard
can be evaluated over these profiles. Moreover, inverted postings over active
rare coordinates allow candidates with no shared evidence to be filtered
before explicit similarity evaluation. Thus, RareSense combines a
pattern-derived representation with established principles of sparse
set-similarity retrieval.

\subsection{Rare-pattern mining and pattern-based anomaly detection}

Pattern mining traditionally focuses on discovering combinations of items that
occur together in transactional data. Frequent-pattern mining, beginning with
association-rule discovery~\cite{agrawal1993rules} and scalable algorithms such
as FP-Growth~\cite{han2000fpgrowth}, emphasizes recurring co-occurrences above
a minimum-support threshold. Such patterns are valuable for summarizing
dominant regularities, but frequent structures are not necessarily the most
informative ones for anomaly-oriented retrieval.

Rare-pattern mining reverses this perspective by targeting low-support
combinations~\cite{szathmary2007rare}. Rare itemsets can expose unusual
co-occurrences that would be removed by conventional frequent-pattern
thresholding. Related work has investigated mechanisms for identifying
sporadic or rare rules~\cite{koh2006ararm} and for improving the efficiency of
rare-itemset discovery~\cite{troiano2009discovering}. These studies establish
that low-frequency combinations can carry information that is not visible from
marginal feature frequencies alone. However, rarity by itself does not
necessarily imply usefulness: an extremely infrequent pattern may reflect
noise, accidental co-occurrence, or an unstable combination. Measures such as
confidence and lift can therefore complement support by indicating whether the
items forming a rule exhibit reliable dependence.

Existing pattern-based anomaly-detection methods exploit pattern statistics
primarily to assign a scalar abnormality score. FPOF
~\cite{he2005fpof}, for example, relates an object's abnormality to the frequent
patterns in which it participates. Other categorical or distribution-oriented
methods operate at the feature level. Attribute Value Frequency
(AVF)~\cite{koufakou2007avf} evaluates an object through the empirical
frequencies of its observed categorical values. HBOS
~\cite{goldstein2012hbos} models feature-wise histograms, ECOD
~\cite{li2022ecod} uses empirical cumulative distribution functions, and COPOD
~\cite{li2020copod} models multivariate tail probabilities through a
copula-based formulation. Although their assumptions differ, their common
output is a query-independent scalar anomaly score.

RareSense uses rare-pattern mining for a fundamentally different purpose.
The mined structures are not merely intermediate statistics used to increase
or decrease an anomaly score. Instead, they become the \emph{coordinates of a
new similarity space}. Minimal rare structures identify unusual
co-occurrences, while retained rules statistically qualify those structures.
An object activates a symbolic coordinate when it contains the evidence
associated with the corresponding rule. Similarity can then be computed
between two objects according to which rare coordinates they jointly activate.

This distinction has several consequences. First, rarity is transformed from a
unary property---``how unusual is this object?''---into a pairwise relation:
``which unusual structures do these two objects share?'' Second, the
representation separates pattern discovery from retrieval: the rare-pattern
dictionary can be mined once and then reused across many queries. Third, the
same coordinates provide human-readable evidence for the resulting similarity.
RareSense therefore connects rare-pattern mining with similarity search rather
than using rare patterns only as ingredients of a conventional detector.

Another important difference concerns the role of individual rare itemsets.
In RareSense, rare itemsets serve primarily as intermediate structures for
discovering candidate higher-order evidence. The final retrieval coordinates
are defined by statistically qualified rule evidence rather than by retaining
every rare conjunction as an independent feature. This avoids simply expanding
the original transaction space with a potentially redundant collection of
itemsets and instead constructs a compact symbolic representation in which
rarity, confidence, lift, and structural complexity can contribute to the
importance of each coordinate.

\subsection{Anomaly detection, ranking, and retrieval}

Unsupervised anomaly detection has produced a broad range of paradigms.
Local-density approaches such as LOF~\cite{breunig2000lof} identify observations
whose neighborhoods differ from those of nearby points. One-class methods
estimate the support of the normal data distribution
~\cite{scholkopf2001svm}, while isolation-based methods detect observations
that can be separated rapidly through recursive partitioning
~\cite{liu2008iforest}. Distance-based formulations rank observations according
to their distance from neighboring or reference objects
~\cite{ramaswamy2000knn,aggarwal2001outlier}.

More recent work has extended anomaly detection through richer statistical and
learned representations. Data-depth methods characterize how centrally or
peripherally an observation lies in a multivariate distribution
~\cite{Mozharovskyi25}. Generative approaches such as ALGAN
~\cite{Bashar25} exploit adversarial learning for time-series anomaly
detection. Deep Isolation Forest~\cite{xu2023dif} combines learned
representations with isolation-based detection, while LUNAR
~\cite{goodge2022lunar} learns to unify local outlier signals through graph
neural networks. These approaches improve the capacity to model nonlinear or
complex data distributions, but their principal objective remains the
estimation of an anomaly score for each individual observation.

This distinction between \emph{anomaly ranking} and \emph{anomaly retrieval} is
important. A detector answers a unary question,
\[
\text{``How anomalous is }x\text{?''}
\]
whereas similarity retrieval answers a pairwise, query-dependent question,
\[
\begin{gathered}
\text{``Which objects are most similar}\\
\text{to a given query } q\text{?''}
\end{gathered}
\]
A highly anomalous object is not necessarily the most relevant neighbor of
another anomaly. Two anomalies may originate from entirely different
mechanisms, while a moderately ranked candidate may share precisely the rare
behavioral structure exhibited by the query. Consequently, sorting candidates
by a global detector score is not equivalent to query-conditioned similarity
search.

RareSense explicitly supports both views while keeping them conceptually
separate. Its rare-pattern representation induces a global rarity-based score
that can be used for conventional anomaly ranking, but its primary retrieval
mechanism compares the rare profiles of a query and candidate directly. This
makes it possible to retrieve anomalies that are \emph{structurally related to
the query}, rather than simply returning the objects with the largest global
outlier scores.

In our evaluation, conventional detectors are therefore included as important
secondary operational references rather than treated as like-for-like
similarity functions. Their rankings indicate how well a user would perform by
inspecting globally suspicious objects, whereas Jaccard-family similarities and
RareSense directly answer a query-conditioned retrieval task.

The distinction is especially relevant in applications where analysts already
possess one suspicious example and seek related cases. Examples include
incident investigation, retrospective threat hunting, fraud analysis, and
diagnostic case retrieval. In such settings, the practical objective is often
not to inspect the entire anomaly ranking but to prioritize a small set of
candidates most relevant to the current case. The anomaly-detection literature
has long emphasized the importance of ranking quality under limited inspection
budgets~\cite{chandola2009survey,pang2021deep,emmott2015meta}. This motivates
our emphasis on top-$k$ retrieval and nDCG@10: relevant anomalous objects should
not merely appear somewhere in the ranking, but should be concentrated near
the top where they can realistically be inspected.

RareSense contributes to this setting by unifying anomaly evidence and
similarity evidence in one representation. Unlike a pipeline in which an
anomaly detector first selects objects and an unrelated similarity measure is
then applied afterward, the same rare coordinates that characterize atypical
structure also determine pairwise relatedness. The retrieval score therefore
has a direct semantic connection to the evidence that makes the objects
unusual.

\subsection{Explainability and symbolic retrieval}

Explainability methods generally seek to clarify the behavior of predictive
models. Post-hoc approaches such as SHAP~\cite{lundberg2017shap} assign feature
contributions to individual predictions, while LIME~\cite{ribeiro2016lime}
approximates a complex predictor locally using an interpretable surrogate.
Counterfactual explanations instead identify changes to an input that would
alter a model decision~\cite{wachter2017counterfactual}. These methods address
important questions such as why a classifier produced a given prediction or
what would need to change for its decision to be different.

Our explainability objective is different. RareSense does not primarily ask why
an individual object received a particular anomaly label. Instead, it explains
\emph{why a particular candidate was retrieved as similar to a particular
query}. This distinction matters because pairwise retrieval explanations
require identifying evidence shared by two objects rather than attributing a
single model output to features.

The RareSense similarity is constructed directly from the weighted overlap of
activated rare-rule coordinates. Consequently, each non-zero contribution to
the similarity score corresponds to explicit symbolic evidence shared by the
query and candidate. A retrieval result can therefore be accompanied by the
specific rare structures responsible for its similarity, together with their
weights and associated statistics. The explanation is thus intrinsic to the
retrieval computation rather than generated afterward by a separate surrogate
model.

This property also differs from explanations based solely on shared atomic
features. Reporting that two objects both contain an individual attribute may
be insufficient when that attribute is common in the data. RareSense instead
can expose higher-order evidence such as a rare conjunction or statistically
qualified rule that is jointly activated by both objects. The explanation
therefore reflects the same higher-order structure that determined the
retrieval score.

Because explanation and retrieval share the same representation, no additional
post-hoc explainer is required, and there is no discrepancy between the
features used to compute similarity and those shown to the analyst. This
provides a form of \emph{explanation-by-construction}: the retrieved objects,
their similarity values, and their symbolic explanations are all derived from
the same rare-pattern dictionary.

\subsection{Positioning of RareSense}

The preceding literature reveals a gap between four established research
directions. Classical similarity measures provide efficient pairwise comparison
but generally operate on predefined atomic features. 

Rare-pattern mining
discovers unusual higher-order structures but is typically used for pattern
discovery or scalar anomaly scoring. Conventional anomaly detectors provide
powerful global rankings but do not naturally define query-conditioned
relatedness between anomalous objects. Finally, post-hoc explainability methods
can interpret model outputs but are external to the similarity computation
itself.

RareSense bridges these directions by treating mined rare structures as a
\emph{similarity representation}. Its main distinction from existing approaches
can be summarized as
\[
\begin{gathered}
\text{rare-pattern discovery}\\
\downarrow\\[-2pt]
\text{symbolic coordinates}\\
\downarrow\\[-2pt]
\text{weighted similarity and retrieval}
\end{gathered}
\]
This design yields three capabilities from the same representation:
(i) query-conditioned retrieval based on shared rare evidence,
(ii) global anomaly ranking from aggregate rare-pattern evidence, and
(iii) intrinsic symbolic explanations obtained directly from the coordinates
contributing to the similarity.

Accordingly, the novelty of RareSense does not lie simply in proposing another
rare-pattern score or another weighted Jaccard variant. The key contribution is
the \emph{construction of a rarity-aware similarity space}: higher-order,
statistically qualified rare evidence replaces individual observed attributes
as the semantic basis of comparison. This makes similarity itself sensitive to
unusual structure and turns rare-pattern mining from an anomaly-scoring
mechanism into a reusable foundation for explainable similarity search.
\section{Problem Setting}
\label{sec:problem}

Let $X=\{x_1,\ldots,x_n\}$ be a collection of objects.
An object may be a network flow, an IoT traffic record, a system-log
session, a process-action row, or a discretized tabular instance.
The only input assumption is that each object can be represented as a
finite transaction over a discrete alphabet.

\subsection{Transactional Input Assumption}

Each object $x_i$ is converted into a transaction
\begin{equation}
    T_i \subseteq \A,
\end{equation}
where $\A=\{a_1,\ldots,a_m\}$ is an alphabet of atoms.
The transactional database is
\begin{equation}
    \T=\{T_1,\ldots,T_n\}.
\end{equation}
The representation covers the following common cases:
\begin{itemize}
    \item \textbf{Binary event vectors}: active columns become atoms.
    \item \textbf{One-hot categorical records}: active attribute--value
    indicators become atoms.
    \item \textbf{Numerical tabular records}: numerical attributes are
    discretized into bins, and bin identifiers become atoms.
    \item \textbf{Log sessions}: event templates, template counts, or
    template transitions become atoms.
\end{itemize}
Labels $y_i\in\{0,1\}$, where $1$ denotes anomaly, are used only for
evaluation.
Label columns and attack-category columns are excluded from $T_i$.

The support of an itemset $p\subseteq\A$ is
\begin{equation}
    \supp(p)=\frac{1}{n}\sum_{i=1}^{n}\mathbf{1}[p\subseteq T_i].
    \label{eq:support}
\end{equation}

Table~\ref{tab:notation} summarizes the main notation used throughout
the paper.

\begin{table*}[t]
\centering
\caption{Key notation used in the RareSense formulation. Minimal rare
itemsets are used only as intermediate seeds for rule generation,
whereas the retained rare rules define the final similarity
coordinates.}
\label{tab:notation}
\small
\begin{tabular}{ll}
\toprule
Symbol & Meaning \\
\midrule
$\A=\{a_1,\ldots,a_m\}$
& Alphabet of behavioral atoms \\

$T_i\subseteq\A$
& Transaction associated with object $x_i$ \\

$\T=\{T_1,\ldots,T_n\}$
& Transactional database \\

$\supp(p)$
& Empirical support of itemset $p\subseteq\A$ \\

$\mathcal{I}_{\tau}$
& Minimal rare itemsets mined under support ceiling $\tau$ \\

$r:b\Rightarrow h$
& Association rule with antecedent $b$ and consequent $h$ \\

$E(r)=b\cup h$
& Complete evidence set associated with rule $r$ \\

$\Rset=\{r_1,\ldots,r_M\}$
& Retained positive-weight rare-rule dictionary \\

$\phi_r(x_i)$
& Indicator that object $x_i$ activates rule coordinate $r$ \\

$P_i\subseteq\Rset$
& Rare-rule profile of object $x_i$ \\

$s(r)$
& Empirical stability of rule $r$ across resamples \\

$w(r)$
& Composite weight of retained rule $r$ \\

$S_R(x_i,x_j)$
& Rare-rule weighted-Jaccard similarity \\

$D_R(x_i,x_j)$
& Induced distance $1-S_R(x_i,x_j)$ \\

$A_R(x_i)$
& Total weighted rare-rule evidence score \\

\bottomrule
\end{tabular}
\end{table*}

\subsection{Primary task: query-by-anomaly retrieval}

Let $\X=\{x_1,\ldots,x_n\}$ denote the collection of candidate objects
represented as transactions. Given a confirmed anomaly query $q$, the search task is to return the subset of $k$
objects most similar to $q$:
\begin{equation}
    \TopK(q)=
    \argmax_{\mathcal K\subseteq X\setminus\{q\},\,|\mathcal K|=k}
    \sum_{x_j\in\mathcal K} S_R(q,x_j),
    \label{eq:topk}
\end{equation}
where $S_R$ is the rarity-aware similarity function defined in
Section~\ref{sec:similarity}. With several confirmed
queries $Q^+$, a candidate can be ranked by
\begin{equation}
    R(x_j\mid Q^+)=\max_{q\in Q^+}S_R(q,x_j).
    \label{eq:qba}
\end{equation}

This evaluates whether the similarity model retrieves objects that
are anomalous in a similar way to known seeds.

Our primary benchmark uses binary relevance: any labeled anomaly is relevant.
Consequently, the primary task evaluates the operational anomaly-search question ``does an anomaly query
bring other anomalies to the top of the list?''.

\subsection{Secondary task: global anomaly ranking}

The secondary task ranks the objects in $\X$ according to a
query-independent anomaly score $A(x_i)$. RareSense instantiates this
score as the total weighted rare-pattern evidence activated by an
object, as formally defined in Section~\ref{sec:global_score}. This task
evaluates whether objects containing larger amounts of strongly weighted
rare evidence tend to be anomalous.

The global-ranking task is complementary to query-conditioned retrieval.
The former asks how much rare evidence an individual object contains,
whereas the latter asks whether two objects activate the same rare
evidence. We therefore evaluate global ranking separately using AUROC.

\section{Rarity-Aware Similarity}
\label{sec:similarity}

This section defines the core contribution, as summarized in
Figure~\ref{fig:raresense_framework}.
Classical similarities compare objects through shared observed atoms.
\method{} compares objects through shared rare explanatory patterns.
\begin{figure*}[ht!]
    \centering
    \includegraphics[width=0.55\linewidth]{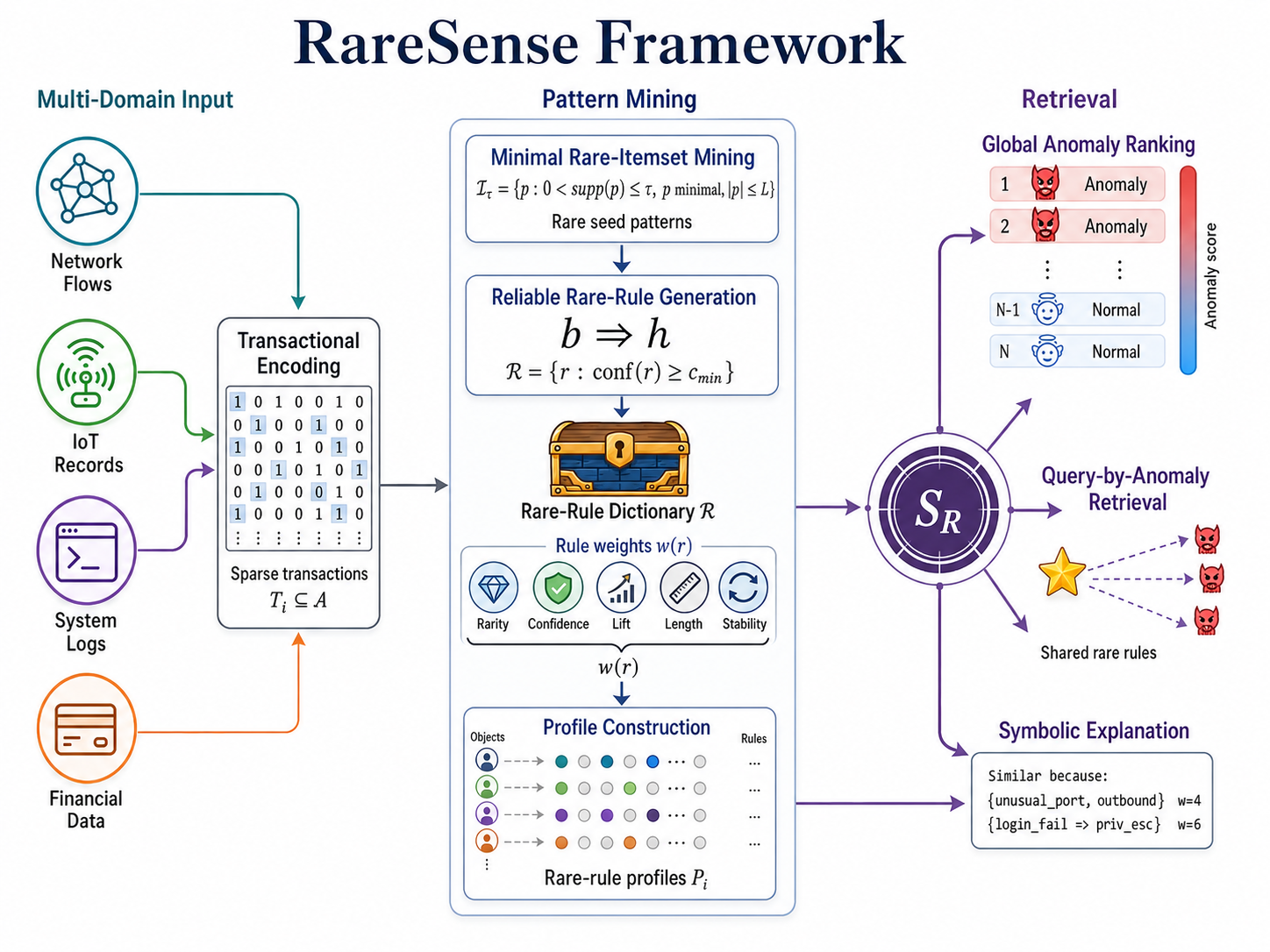}
    \caption{Overview of the \method{} framework.
    Input objects such as network flows, IoT traffic records, or log
    sessions are represented as sparse transactions.
    Rare itemsets are mined under support and length constraints as an
    intermediate substrate; reliable rare association rules are generated
    from them and weighted by rarity, confidence, lift, length, and
    stability to build rare-rule profiles.
    \method{} then computes rarity-aware similarity over shared
    rare-rule evidence to support global anomaly ranking,
    query-by-anomaly retrieval, and symbolic explanation through shared
    retained rules.}
    \label{fig:raresense_framework}
\end{figure*}

\paragraph{Motivation.}
IDF-weighted Jaccard~\cite{salton1988term} is a strong baseline that
corrects uniform matching by downweighting frequent atoms.
However, it remains atom-wise: individually common atoms may form a
rare and meaningful conjunction that IDF cannot detect, while
individually rare atoms may be unrelated to the anomaly mechanism.
\method{} therefore \emph{extends IDF weighting from atoms to patterns}.

\subsection{Rare Behavioral Rule Dictionary}

A non-empty itemset $p$ is \emph{rare} under threshold $\tau$ when
$0<\supp(p)\le\tau$.  We mine \emph{minimal rare itemsets} (MRIs):
\begin{equation}
\begin{aligned}
\mathcal I_\tau=\{p\subseteq\A:\;&
0<\supp(p)\le\tau,\\
&\supp(q)>\tau\;\;\forall q\subsetneq p,\quad
|p|\le L
\}.
\end{aligned}
\label{eq:mri}
\end{equation}
Minimality removes redundant rare supersets whose rarity is already explained
by a smaller rare subset and makes the mining stage more compact. The parameter ($L$) denotes the maximum itemset cardinality, preventing the enumeration of overly long and highly specific patterns. These rare itemset candidates are intermediate structures; they are not retained directly as similarity coordinates.

Each $p\in\mathcal I_\tau$ seeds candidate association rules
$r:b\Rightarrow h$. For partition rules, $b\cup h=p$,
$b\cap h=\emptyset$, and both sides are non-empty. In the reported
implementation we generate both non-trivial partition rules and exact closure
rules derived from the closure of $p$. Consequently, a closure-rule evidence
set can be longer than the MRI seed even though MRI enumeration is bounded by
$L$. For every rule,
\begin{equation}
\begin{aligned}
\supp(r) &= \supp(b\cup h), \qquad
\conf(r) = \frac{\supp(b\cup h)}{\supp(b)},\\
\lift(r) &= \frac{\supp(b\cup h)}
{\supp(b)\supp(h)}.
\end{aligned}
\label{eq:rule_stats}
\end{equation}
Rules with $\conf(r)\ge c_{\min}$ are retained.  Lift is used as a dependence
quality in the weight rather than as a hard requirement, avoiding brittle
thresholding when support is extremely small.

Let $\Rset=\{r_1,\ldots,r_M\}$ be the retained \emph{RareSense dictionary} and define the
\emph{evidence set}
\begin{equation}
    E(r)=b\cup h.
\end{equation}
Importantly, a rule coordinate is active only when its \emph{complete evidence
set} is present:
\begin{equation}
    \phi_r(x_i)=\mathbf{1}[E(r)\subseteq T_i].
    \label{eq:activation}
\end{equation}
This is intentionally different from antecedent-only rule firing.  The rule is
used as a reliability-qualified conjunction; confidence and lift determine the
quality of the conjunction, while $E(r)$ determines activation.

The rare-rule profile of an object $x_i$ is therefore
\begin{equation}
    P_i=\{r\in\Rset:E(r)\subseteq T_i\}.
    \label{eq:profile}
\end{equation}

\subsection{From Transactional Feature Space to RareSense Space}
\label{sec:raresense_space}
Let $\X=\{x_1,\ldots,x_n\}$ denote the same set of original objects in the
dataset and let $\Rset=\{r_1,\ldots,r_M\}$ be the retained rare-rule
dictionary. RareSense maps each object from the original object space to
an $M$-dimensional binary rare-rule activation space:
\begin{equation}
\begin{aligned}
\Phi_R &: \X \rightarrow \{0,1\}^{M},\\
\Phi_R(x_i) &=
\bigl(\phi_{r_1}(x_i),\ldots,\phi_{r_M}(x_i)\bigr).
\end{aligned}
\label{eq:mapping}
\end{equation}
Each rule receives a non-negative weight
\begin{equation}
\begin{aligned}
w(r)={}&
\left[\log\frac{n+\eta}{n_r+\eta}\right]^{\alpha}
\times[\conf(r)]^{\beta}\\
&\times
[1+\log(\max\{1,\lift(r)\})]^{\gamma}\\
&\times
[1+\log(1+|E(r)|)]^{\delta}
\times[s(r)]^{\zeta}.
\end{aligned}
\label{eq:weight}
\end{equation}
where $n_r$ is the number of objects activating rule $r$, $\eta>0$ is a
smoothing constant, and $s(r)$ denotes the stability of the rule across
resamples. The exponents $\alpha$, $\beta$, $\gamma$, $\delta$, and $\zeta$
control the contributions of rarity, confidence, lift, rule complexity, and
stability, respectively. The stability term
\begin{equation}
    s(r)=\frac{1}{B}\sum_{b=1}^{B}\mathbf{1}[r\in\mathcal{R}^{(b)}]
    \label{eq:stability}
\end{equation}
measures the empirical rediscovery rate across subsamples. In the
reported experiments, we use $B=10$ subsamples and set $\zeta=1$.
Consequently, stability enters the weight linearly: a rule rediscovered
in $k$ of the ten subsamples has $s(r)=k/10$, and its weight is
multiplied by this empirical rediscovery rate. Rules with zero empirical stability have zero composite weight and
are omitted from the profiles; hence all retained coordinates satisfy
$w(r)>0$.

A weighted coordinate embedding is
\begin{equation}
\Psi_R(x_i)=
\bigl(\sqrt{w(r_1)}\phi_{r_1}(x_i),\ldots,
      \sqrt{w(r_M)}\phi_{r_M}(x_i)\bigr).
\label{eq:weighted_embedding}
\end{equation}
The square root is convenient because two active copies of the same coordinate
contribute exactly $w(r)$ to their inner product.

Figure~\ref{fig:raresense_space} illustrates the RareSense space transformation.
\begin{figure*}
    \centering
    \includegraphics[width=0.7\linewidth]{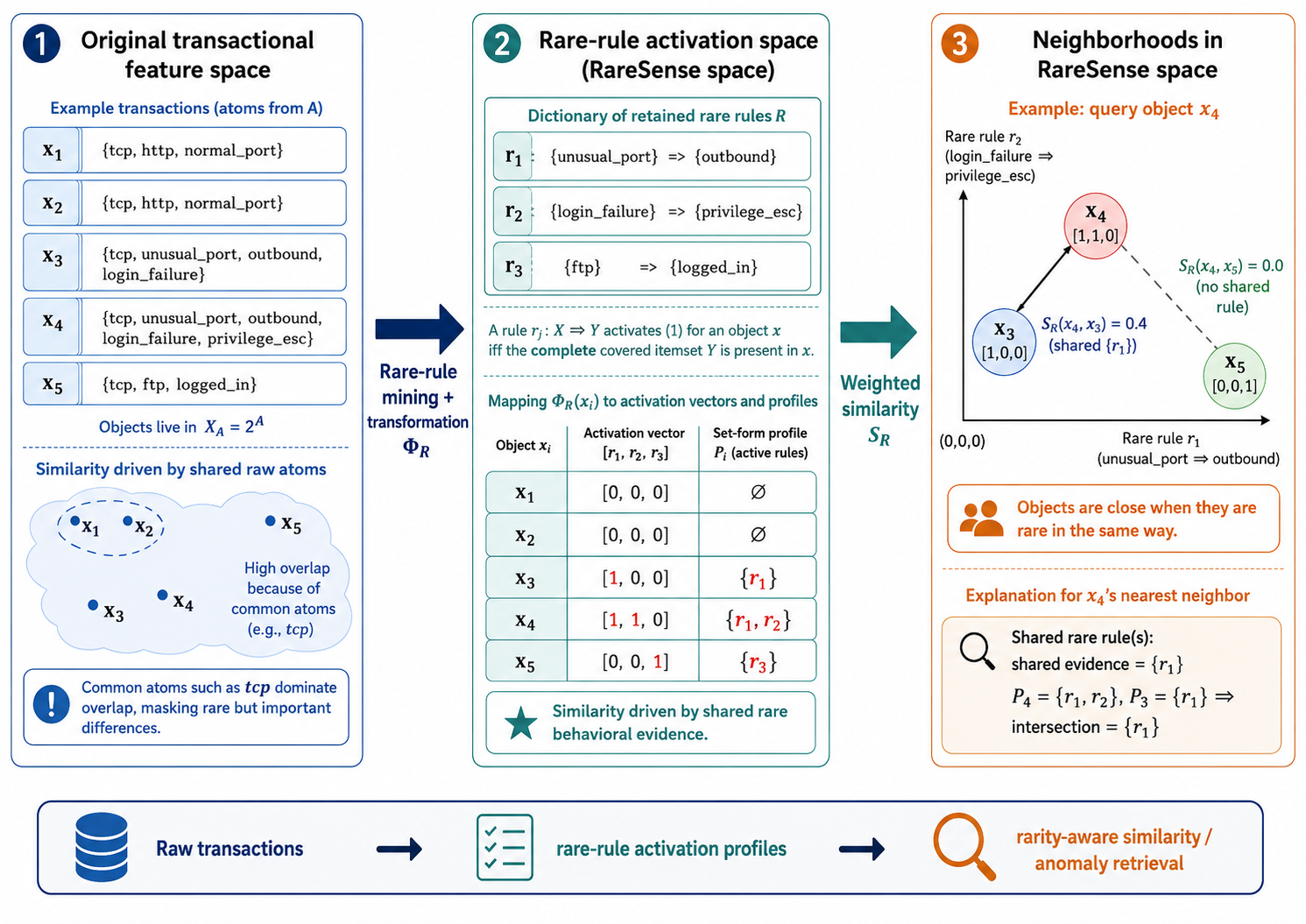}
    \caption{RareSense space transformation.
    Objects are mapped from the original transactional feature space,
    where similarity is driven by shared raw atoms, to a rare-rule
    activation space through $\Phi_R$.
    In the transformed space, each coordinate corresponds to a retained
    rare rule activated by the object, and weighted similarity compares
    objects according to shared rare behavioral evidence.}
    \label{fig:raresense_space}
\end{figure*}

\subsection{Rare-rule weighted Similarity}

For two profiles,
\begin{equation}
S_R(x_i,x_j)=
\begin{cases}
\dfrac{\sum_{r\in P_i\cap P_j}w(r)}
      {\sum_{r\in P_i\cup P_j}w(r)},
& P_i\cup P_j\neq\emptyset,\\[8pt]
1,&P_i=P_j=\emptyset.
\end{cases}
\label{eq:raresense_similarity}
\end{equation}
The numerator is shared rare evidence; the denominator is all rare evidence
activated by either object. The value $S_R(\emptyset,\emptyset)=1$ is required
for the metric statement on profile equivalence classes. Operational retrieval
uses a separate no-evidence policy: a query with $P_q=\emptyset$ is not treated
as evidence of semantic similarity merely because another object also has an
empty profile. In the benchmark, such queries are retained rather than removed;
all candidates receive zero evidence-based similarity and ties are resolved by
a fixed deterministic candidate order. This prevents optimistic filtering of
hard queries while keeping the theoretical similarity definition intact.

\subsection{Running Example: Rare Scoring vs.\ Rare Similarity}
\label{sec:running_example}

We illustrate the distinction with five network connections:
\begin{table}[h]
\centering
\caption{Illustrative transactional objects and their observed atoms.}
\label{tab:toy_example}
\small
\setlength{\tabcolsep}{4pt}
\begin{tabularx}{\columnwidth}{l>{\raggedright\arraybackslash}X}
\toprule
Object & Observed atoms\\
\midrule
$x_1$ &
\texttt{tcp}, \texttt{http}, \texttt{normal\_port}\\

$x_2$ &
\texttt{tcp}, \texttt{http}, \texttt{normal\_port}\\

$x_3$ &
\texttt{tcp}, \texttt{unusual\_port}, \texttt{outbound},
\texttt{login\_failure}\\

$x_4$ &
\texttt{tcp}, \texttt{unusual\_port}, \texttt{outbound},
\texttt{login\_failure}, \texttt{privilege\_esc}\\

$x_5$ &
\texttt{tcp}, \texttt{ftp}, \texttt{logged\_in}\\
\bottomrule
\end{tabularx}
\end{table}
Suppose the retained rare rules are
$r_1:\{\texttt{unusual\_port}\}\Rightarrow\{\texttt{outbound}\}$
with $w(r_1)=4$,
$r_2:\{\texttt{login\_failure}\}\Rightarrow\{\texttt{privilege\_esc}\}$
with $w(r_2)=6$, and
$r_3:\{\texttt{ftp}\}\Rightarrow\{\texttt{logged\_in}\}$ with $w(r_3)=2$.
Then $P_3=\{r_1\}$, $P_4=\{r_1,r_2\}$, and $P_5=\{r_3\}$:
\[
S_R(x_4,x_3)=\frac{4}{4+6}=0.4,\qquad
S_R(x_4,x_5)=0.
\]
Thus $x_5$ contains rare evidence but not the same rare evidence as $x_4$.
This is the distinction between pointwise rarity and pairwise rare similarity.

\subsection{Formal properties}
\label{sec:properties}
\begin{proposition}[IDF-Jaccard as a singleton-coordinate special case]
Consider the RareSense profile-similarity construction with a
coordinate dictionary consisting only of singleton evidence
coordinates
\[
\mathcal R_{\mathrm{sing}}
=
\{r_a : E(r_a)=\{a\},\ a\in\mathcal A\},
\]
and assign
\[
w(r_a)=\log\frac{n+\eta}{n_a+\eta}.
\]
These formal singleton coordinates are introduced only for the
reduction and need not be generated by the rare-rule mining procedure.
Then the RareSense similarity reduces exactly to IDF-weighted Jaccard.
\end{proposition}

\begin{proposition}[Weighted-Jaccard pseudometric]
\label{prop:pseudometric}
Let the rare-rule dictionary $\mathcal R$ be fixed and let
$w(r)>0$ for every $r\in\mathcal R$.
Then
\[
D_R(x_i,x_j)=1-S_R(x_i,x_j)
\]
is a pseudometric on the original object collection $X$.
Moreover, it is a metric on the quotient space $X/\!\sim$, where
\[
x_i\sim x_j
\quad\Longleftrightarrow\quad
P_i=P_j.
\]
\end{proposition}

\begin{proof}
Let
\[
W(A)=\sum_{r\in A}w(r).
\]
For $P_i\cup P_j\neq\emptyset$,
\[
D_R(x_i,x_j)
=
1-
\frac{W(P_i\cap P_j)}
     {W(P_i\cup P_j)}
=
\frac{W(P_i\triangle P_j)}
     {W(P_i\cup P_j)},
\]
which is the weighted Jaccard distance.

Non-negativity and symmetry follow immediately from the corresponding
properties of set intersection and union. The triangle inequality follows
from the metric property of weighted Jaccard distance for fixed non-negative
coordinate weights
~\cite{marczewski1958certain,levandowsky1971distance,lipkus1999tanimoto}.

If $P_i=P_j=\emptyset$, Eq.~\eqref{eq:raresense_similarity} defines
$S_R(x_i,x_j)=1$, and therefore $D_R(x_i,x_j)=0$, consistently extending
the distance to the empty profile.

For every object, $D_R(x_i,x_i)=0$. However, two distinct original objects
may activate exactly the same rare-rule profile, so
\[
D_R(x_i,x_j)=0
\]
does not necessarily imply $x_i=x_j$. Hence $D_R$ is a pseudometric on
$X$. Since $D_R$ depends only on the corresponding rare-rule profiles,
it is well defined on equivalence classes induced by
$x_i\sim x_j\Leftrightarrow P_i=P_j$. On this quotient space,
identity of indiscernibles is restored, and $D_R$ is therefore a metric
on $X/\!\sim$.
\end{proof}
\begin{proposition}[Higher-order tie breaking]
\label{prop:tiebreak}
There exist object pairs that are tied by an atom-level weighted Jaccard
similarity but strictly separated by RareSense.
\end{proposition}

\begin{proof}
Let all four atoms $a,b,c,d$ have equal positive atom weight and define
$q=\{a,b,c\}$, $x^+=\{a,b,d\}$, and $x^-=\{a,c,d\}$.
Both candidates share two of four union atoms with $q$, hence atom-level
weighted Jaccard gives the same value $1/2$.
Now suppose the retained rare-rule dictionary contains a coordinate $r^*$ with
$E(r^*)=\{a,b\}$ and no other retained coordinate is active in these three
objects.  Then
$P_q=P_{x^+}=\{r^*\}$ and $P_{x^-}=\emptyset$, so
$S_R(q,x^+)=1$ and $S_R(q,x^-)=0$.
Thus a higher-order coordinate can resolve an atom-level tie.  The proposition
is existential; it does not claim that RareSense universally dominates atomic
similarity.
\end{proof}

\begin{proposition}[Exact explanation decomposition]
\label{prop:explanation}
For $P_i\cup P_j\neq\emptyset$, define the contribution of a shared rule
$r\in P_i\cap P_j$ by
\begin{equation}
c_r(x_i,x_j)=
\frac{w(r)}{\sum_{u\in P_i\cup P_j}w(u)}.
\label{eq:rule_contrib}
\end{equation}
Then
\begin{equation}
S_R(x_i,x_j)=\sum_{r\in P_i\cap P_j}c_r(x_i,x_j).
\end{equation}
\end{proposition}

\begin{proof}
Immediate by distributing the denominator of
Eq.~\eqref{eq:raresense_similarity} over the shared-rule numerator.
\end{proof}

\subsection{Compatibility with metric and inverted indexing}
\label{sec:metric_indexing}

Proposition~\ref{prop:pseudometric} makes exact metric indexing possible on the
\emph{profile space}.  VP-trees~\cite{yianilos1993vptree} and
M-trees~\cite{ciaccia1997mtree} can therefore use $D_R$ for exact pruning, with
objects sharing the same profile occupying zero-distance equivalence classes.
This compatibility is a property of the similarity geometry; the present
experiments use direct evaluation rather than a specialized metric index.

Rare-rule profiles are also naturally sparse.  An inverted index can store, for
each rule $r$, the posting list of objects that activate it.  For a query $q$,
the union of postings for $r\in P_q$ forms an exact candidate set for all objects
with non-zero RareSense similarity.  Candidates outside that set have zero
shared rare evidence and need not be scored unless zero-similarity ties must be
materialized.  This filtering is independent of the triangle inequality and can
be combined with a metric index.

For approximate search, weighted MinHash~\cite{ioffe2010improved} provides
sketches for weighted Jaccard, while LSH-style bucketing
~\cite{broder1997resemblance,indyk1998ann} can be used to reduce candidate
generation further.  These mechanisms suggest a two-stage deployment:
inverted or sketch-based filtering followed by exact RareSense re-ranking and
rule-level explanation.

\subsection{Explanation of retrieved neighbors}
\label{sec:explain}

Because the similarity is additive over shared rules, explanations are exact
rather than post hoc.  We define
\begin{equation}
\operatorname{Explain}_m(x_i,x_j)=
\operatorname{Top}_m
\left\{
(r,c_r(x_i,x_j)):
r\in P_i\cap P_j
\right\},
\label{eq:explain}
\end{equation}
where $c_r$ is given by Eq.~\eqref{eq:rule_contrib}.  Each explanation reports
the highest-contributing shared rules together with their normalized
contributions.

The fidelity of a truncated $m$-rule explanation is
\begin{equation}
F_m(x_i,x_j)=
\frac{\sum_{r\in\operatorname{Explain}_m(x_i,x_j)}w(r)}
     {\sum_{r\in P_i\cap P_j}w(r)},
\label{eq:explain_fidelity}
\end{equation}
with $F_m=1$ when the displayed rules account for all shared evidence.  In the
running example, the pair $(x_4,x_3)$ has only one shared rule, so its
one-rule explanation has fidelity $1$ and contributes the full similarity
$4/(4+6)=0.4$.

\subsection{Secondary global anomaly score}
\label{sec:global_score}

For global ranking we use the total weighted rare evidence
\begin{equation}
A_R(x_i)=\sum_{r\in P_i}w(r).
\label{eq:anomaly_score}
\end{equation}
This score answers a different question from $S_R$: $A_R$ asks how much strong
rare evidence an object contains, whereas $S_R$ asks whether two objects share
the same evidence.  The reported AUROC results use Eq.~\eqref{eq:anomaly_score};
no neighborhood-isolation term is included in the main experiments.

\subsection{Adaptive rarity-scale selection}
\label{sec:auto_config}

The support ceiling controls a fundamental rarity--coverage trade-off. Very
small $\tau$ values retain highly selective evidence but can make rule profiles
too sparse for query-conditioned retrieval; broader ceilings increase profile
overlap but admit less exceptional patterns. Unless explicitly labeled
\methodfixed{}, \method{} therefore denotes the adaptive configuration below.

First, a preliminary model is fitted using
$\tau_0=0.01$, $c_{\min}=0.95$, and $L_0=4$. For each object we compute a
lightweight preliminary rarity statistic
\begin{equation}
    s_i^{(0)}=\sum_{r\in P_i^{(0)}}\bigl(1-\supp(r)\bigr),
    \label{eq:prelim_rarity}
\end{equation}
where $P_i^{(0)}$ is the preliminary rule profile. We use the empirical
upper-tail mass as an unlabeled proxy for the amount of atypical structure:
\begin{equation}
  \hat\rho=
  \operatorname{clip}\!\left(
    \frac{\left|\left\{i:s_i^{(0)}>\bar{s}^{(0)}+2\sigma_s^{(0)}\right\}\right|}{n},
    \;0.001,\;0.10
  \right),
  \label{eq:rho_hat}
\end{equation}
where $\bar{s}^{(0)}$ and $\sigma_s^{(0)}$ are the mean and standard
deviation of the preliminary statistics. The quantity $\hat\rho$ is used as
an operational proxy rather than claimed to be an unbiased estimate of the
true anomaly prevalence.

The final support ceiling is
\begin{equation}
  \tau=\operatorname{clip}(\kappa\hat\rho,\;0.001,\;0.10),
  \label{eq:tau_auto}
\end{equation}

with a fixed global coverage multiplier $\kappa=5$. To avoid confusion
with the weight exponent $\alpha$ in Eq.~\eqref{eq:weight}, we denote
this multiplier by $\kappa$. It is a heuristic calibration constant used
unchanged for every workload and was not optimized using class labels,
validation queries, or a held-out validation set. The workload-specific choice of $\tau$ therefore depends only on the
unlabeled statistic $\hat{\rho}$. The multiplier broadens the rarity
ceiling relative to the estimated atypical tail, allowing patterns shared
by several anomalous objects to remain eligible rather than forcing every
retained pattern to be as rare as the estimated anomaly fraction itself.

The maximum MRI length is adapted as
\begin{equation}
L=
\begin{cases}
5, & \hat\rho\ge 0.05,\\
4, & 0.01\le\hat\rho<0.05,\\
3, & \hat\rho<0.01.
\end{cases}
\label{eq:L_auto}
\end{equation}
The confidence threshold remains fixed at $c_{\min}=0.95$. The adaptive
procedure therefore changes only the rarity scale and maximum MRI length. It
requires one preliminary mining pass followed by the final fit; we treat this
as explicit model-selection overhead rather than assuming it is negligible.

\section{Algorithms}
\label{sec:algorithms}

The complete \method{} workflow separates representation construction,
adaptive rarity-scale selection, query-time retrieval, and global anomaly
ranking. Algorithms~\ref{alg:build}--\ref{alg:global} make this separation
explicit. The first two procedures are executed during model construction:
Algorithm~\ref{alg:auto} determines the workload-specific rarity scale, while
Algorithm~\ref{alg:build} transforms the original transactional representation
into the rare-rule profile space using the resulting parameters. Once the rule
dictionary and object profiles have been constructed, query-by-anomaly retrieval
and global anomaly ranking operate directly on the resulting sparse profiles
without repeating the mining stage. This separation is operationally useful
because the comparatively expensive pattern-mining stage is performed offline,
whereas subsequent similarity queries reuse the precomputed dictionary, weights,
profiles, and posting lists.

\begin{algorithm}
\caption{Construct the RareSense rule space}
\label{alg:build}
\begin{algorithmic}[1]
\Require Objects $X$, transaction mapper $\varphi$, rarity threshold $\tau$,
maximum MRI length $L$, confidence threshold $c_{\min}$
\Ensure Rule dictionary $\Rset$, weights $w$, profiles $\{P_i\}_{i=1}^n$
\ForAll{$x_i\in X$}
    \State $T_i\gets\varphi(x_i)$
\EndFor
\State Mine minimal rare itemsets $\mathcal I_\tau$ satisfying
Eq.~\eqref{eq:mri}
\State Generate partition rules and exact closure rules from
$\mathcal I_\tau$
\State $\Rset\gets\{r:\conf(r)\ge c_{\min}\}$
\ForAll{$r\in\Rset$}
    \State Compute $w(r)$ using Eq.~\eqref{eq:weight}
\EndFor
\ForAll{$x_i\in X$}
    \State $P_i\gets\{r\in\Rset:E(r)\subseteq T_i\}$
\EndFor
\State \Return $\Rset,w,\{P_i\}$
\end{algorithmic}
\end{algorithm}

Algorithm~\ref{alg:build} constructs the common representation used by
all subsequent \method{} operations. Objects are first mapped to
transactions, after which minimal rare itemsets are mined as intermediate
seeds for partition and exact closure rules. Rules satisfying
$\conf(r)\ge c_{\min}$ form the final dictionary, receive the weights in
Eq.~\eqref{eq:weight}, and define each object's sparse profile through
complete-evidence activation. The resulting dictionary, weights, and
profiles are reused for both query-conditioned retrieval and global
anomaly ranking.

\begin{algorithm}
\caption{Adaptive rarity-scale selection for \method{}}
\label{alg:auto}
\begin{algorithmic}[1]
\Require Transaction database $\T=\{T_1,\ldots,T_n\}$,
global coverage multiplier $\kappa$
\Ensure Adaptive parameters $(\tau,L,c_{\min})$
\State Fit preliminary rule space with
$\tau_0=0.01$, $L_0=4$, $c_{\min}=0.95$
\State $s_i^{(0)}\gets
\sum_{r\in P_i^{(0)}}(1-\supp(r))$ for all $i$
\State $\hat\rho\gets\operatorname{clip}\!\left(
\frac{|\{i:s_i^{(0)}>\bar{s}^{(0)}+2\sigma_s^{(0)}\}|}{n},
0.001,0.10\right)$
\State $\tau\gets\operatorname{clip}(\kappa\hat\rho,0.001,0.10)$
\State $L\gets5$ if $\hat\rho\ge0.05$,
else $4$ if $\hat\rho\ge0.01$, else $3$
\State $c_{\min}\gets0.95$
\State \Return $(\tau,L,c_{\min})$
\end{algorithmic}
\end{algorithm}

Algorithm~\ref{alg:auto} addresses the fact that a single fixed rarity
threshold does not generate equally useful profiles across workloads with
different levels of sparsity and atypical structure. The procedure begins with
a preliminary \method{} fit using fixed reference parameters. From the
resulting preliminary profiles, each object receives a lightweight rarity
statistic $s_i^{(0)}$ based on the supports of the rules it activates. The
fraction of objects whose statistic exceeds
$\bar{s}^{(0)}+2\sigma_s^{(0)}$ defines $\hat\rho$, an unlabeled proxy for the
amount of unusually rare structure present in the workload.

The proxy $\hat\rho$ is then translated into the two parameters that most
directly control the richness of the rare-rule space. The support ceiling is
set according to
$\tau=\operatorname{clip}(\kappa\hat\rho,0.001,0.10)$, where $\kappa$ is a
global coverage multiplier, while the maximum MRI length $L$ is increased as
the estimated atypical mass grows. The intuition is that an excessively strict
support ceiling may produce highly selective but nearly disjoint profiles,
which is undesirable for retrieval because related anomalies must activate at
least some common evidence. A broader rarity scale increases the opportunity
for meaningful profile overlap, while the upper clipping bound prevents the
dictionary from drifting too far toward common patterns.

The confidence threshold remains fixed at $c_{\min}=0.95$, so the adaptive
procedure changes only the rarity scale and maximum MRI length. It requires two
passes: a preliminary rule-space construction used to estimate $\hat\rho$,
followed by the final model built with the adapted $(\tau,L)$. The global multiplier is fixed heuristically to $\kappa=5$ and used
unchanged for every workload. Only $\hat\rho$, and hence the
workload-specific values of $\tau$ and $L$, varies across workloads.
No class labels are used in this adaptation.

\begin{algorithm}
\caption{Exact rarity-aware top-$k$ retrieval}
\label{alg:query}
\begin{algorithmic}[1]
\Require Query $q$, precomputed profiles $\{P_i\}$, weights $w$, $k$
\Ensure Top-$k$ neighbors with explanations
\If{$P_q=\emptyset$}
    \State Assign zero evidence-based similarity to every candidate
    \State \Return deterministic zero-similarity ranking (benchmark policy)
\EndIf
\State $C(q)\gets\bigcup_{r\in P_q}\operatorname{Posting}(r)$
\ForAll{$x_j\in C(q)\setminus\{q\}$}
    \State Compute $S_R(q,x_j)$
\EndFor
\State Assign similarity $0$ to candidates outside $C(q)$
\State Keep the $k$ largest similarities,
resolving zero-score ties deterministically
\ForAll{returned $x_j$ with $S_R(q,x_j)>0$}
    \State Compute $\operatorname{Explain}_m(q,x_j)$
\EndFor
\State \Return ranked neighbors and explanations
\end{algorithmic}
\end{algorithm}

Algorithm~\ref{alg:query} implements the primary query-conditioned retrieval
task. Because rare-rule profiles are sparse, an inverted index associates each
rule $r$ with the objects that activate it. For a query $q$, only objects
appearing in at least one posting list associated with $P_q$ can obtain a
positive \method{} similarity. Consequently,
\[
C(q)=\bigcup_{r\in P_q}\operatorname{Posting}(r)
\]
forms an exact candidate set rather than an approximation: any object outside
$C(q)$ shares no retained rule with the query and therefore necessarily has
similarity zero.

Exact weighted-Jaccard similarity is computed only for candidates in $C(q)$,
after which the $k$ highest-scoring objects are returned. This separates
candidate generation from exact re-ranking and avoids unnecessary comparisons
against objects that cannot share rare evidence with the query. Because the
same shared rules that determine the similarity are retained after ranking,
the explanation of each positive-similarity neighbor requires no surrogate
model. The largest rule-level contributions defined in
Eq.~\eqref{eq:rule_contrib} directly explain why that candidate was retrieved.

The empty-profile case requires a separate operational policy. If
$P_q=\emptyset$, the query contains no retained rare evidence from which a
rule-mediated relation to another object can be established. Although two empty
profiles are mathematically assigned similarity one in
Eq.~\eqref{eq:raresense_similarity} to preserve the metric construction,
interpreting empty profiles as meaningful anomaly neighbors would be
operationally misleading. The benchmark therefore assigns zero
evidence-based similarity to all candidates and resolves the resulting ties
deterministically. These queries remain part of the evaluation rather than
being discarded, avoiding optimistic performance estimates caused by excluding
difficult queries.

\begin{algorithm}
\caption{Global anomaly ranking}
\label{alg:global}
\begin{algorithmic}[1]
\Require Profiles $\{P_i\}$ and weights $w$
\Ensure Ranked anomaly list
\ForAll{$x_i\in X$}
    \State $A_R(x_i)\gets\sum_{r\in P_i}w(r)$
\EndFor
\State \Return objects sorted by decreasing $A_R(x_i)$
\end{algorithmic}
\end{algorithm}

Algorithm~\ref{alg:global} uses the same rare-rule representation for the
secondary global anomaly-ranking task. Unlike Algorithm~\ref{alg:query}, which
compares one profile against another and answers the question
``which objects share rare evidence with this query?'', the global score
aggregates the total weighted rare evidence activated by each object.
Objects activating more strongly weighted rare rules therefore receive larger
$A_R(x_i)$ values and are ranked as more anomalous.

This distinction is central to the dual use of the \method{} representation.
The pairwise similarity $S_R$ is query-conditioned and changes according to the
rare evidence present in the query, whereas $A_R$ is a query-independent scalar
score. The two algorithms therefore answer complementary questions using the
same underlying rule space: Algorithm~\ref{alg:query} supports investigation
after a suspicious or confirmed anomaly has been selected, while
Algorithm~\ref{alg:global} supports initial prioritization of suspicious objects
before a query is available. This distinction also motivates the separate
evaluation protocols used later: nDCG@10 for query-conditioned retrieval and
AUROC for global anomaly ranking.

Taken together, the four algorithms define a two-stage operational workflow.
Algorithm~\ref{alg:auto} first selects the workload-specific rarity scale, after
which Algorithm~\ref{alg:build} constructs the final dictionary and sparse
profiles. The resulting representation can then support repeated retrieval
queries through Algorithm~\ref{alg:query}, while Algorithm~\ref{alg:global}
provides a complementary global ranking from the same profiles. Thus the
symbolic mining stage is amortized across subsequent searches, and both
retrieval and explanation reuse the same precomputed rare-rule evidence rather
than requiring separate models.
\section{Computational Analysis}
\label{sec:complexity}

Let $\bar t$ be the mean transaction length, $M_I$ the number of mined MRIs,
$M=|\Rset|$ the number of retained rule coordinates, $\bar e$ the mean rule
evidence-set size, and $\bar P$ the mean profile size.  Mining rare itemsets is
output- and data-dependent and is exponential in the worst case, as is frequent
itemset mining.  With a maximum MRI size $L$, a level-wise implementation that
evaluates candidate family $C_\ell$ has candidate-evaluation cost
\[
O\!\left(
n\sum_{\ell=1}^{L}|C_\ell|\ell
\right).
\]
Rule generation from an MRI has at most $2^{|p|}-2$ non-trivial partitions, so
partition-rule generation costs $O(M_I2^LL)$ before filtering.  Closure-rule
construction adds the cost of computing closures for the mined MRIs.

Once the dictionary is fixed, profile construction is
$O(nM\bar e)$ naively and can be reduced substantially with bitsets or inverted
postings.  A single weighted-Jaccard comparison is
$O(|P_q|+|P_j|)$ for sorted sparse profiles.  A full exact scan costs
$O(n\bar P+n\log k)$, whereas inverted filtering replaces $n$ by the candidate
count $|C(q)|$ plus the cost of reading the posting lists.  Global scoring with
Eq.~\eqref{eq:anomaly_score} is linear in the number of active
object--rule incidences, $O(n\bar P)$.

\begin{table*} [h]
\centering
\caption{Complexity of the main post-mining operations.  Mining itself is
candidate-dependent and worst-case exponential in the atom vocabulary.}
\label{tab:complexity}
\small
\begin{tabular}{lll}
\toprule
Operation & Time & Space\\
\midrule
Transaction construction & $O(n\bar t)$ & $O(n\bar t)$\\
Profile construction (naive) & $O(nM\bar e)$ & $O(n\bar P)$\\
Single similarity & $O(|P_q|+|P_j|)$ & $O(1)$ extra\\
Exact full-scan top-$k$ & $O(n\bar P+n\log k)$ & $O(k)$\\
Inverted candidate retrieval & $O(\sum_{r\in P_q}|\mathrm{Posting}(r)|+
|C(q)|\bar P)$ & postings\\
Global anomaly score & $O(n\bar P)$ & $O(n)$ scores\\
\bottomrule
\end{tabular}
\end{table*}

\section{Experimental Evaluation}
\label{sec:experiments}

\subsection{Datasets}

We evaluate \method{} on 27 anomaly-detection workloads drawn from
four benchmark families. The evaluation combines network-security,
provenance-based, and general categorical benchmarks spanning text,
web, scientific, marketing, image-derived, and bioassay data.
Table~\ref{tab:datasets} reports the processed workload statistics.

\paragraph{UWF-ZeekData24 (14 datasets).}
UWF-ZeekData24\footnote{https://datasets.uwf.edu/}~\cite{elam2025uwf} is a recent and realistic network-flow benchmark
collected from a controlled cyber-range environment.
Each record is a Zeek-parsed connection log entry represented as a
binary vector over 42--43 behavioral atoms (protocol, service, duration
bins, port-range indicators).
Two representation variants are provided:
\emph{Variant~1} mixes benign and attack flows into a single
transactional database across seven ATT\&CK tactics (Credential
Access, Defense Evasion, Exfiltration, Initial Access, Persistence,
Privilege Escalation, Reconnaissance), while \emph{Variant~2} adds
host-side Windows event atoms to the Zeek features.
This yields fourteen datasets of 16{,}530--91{,}400 objects with anomaly
ratios ranging from 0.05\% to 47.59\%.

\paragraph{NSL-KDD (2 datasets).}
We use the two most challenging subsets of the NSL-KDD intrusion
detection benchmark\footnote{https://www.kaggle.com/datasets/hassan06/nslkdd}~\cite{Tavallaee09,Mishra24}: \textsc{Probe} (port scans;
6.43\%; $n=64{,}759$) and \textsc{U2R} (user-to-root privilege
escalation; 0.37\%; $n=60{,}821$).
Categorical connection features are converted to attribute--value atoms;
numerical features are discretized into five quantile bins.

\paragraph{DARPA Transparent Computing (4 datasets).}
We evaluate on four provenance-graph datasets from the DARPA TC \footnote{https://gitlab.com/adaptdata}
engagement~\cite{darpatc,BENABDERRAHMANE2026114877}:  Android Clearscope ($n=102$; 8.82\%),
Windows 5-dir Events ($n=17{,}569$; 0.05\%), Linux Trace
($n=272{,}376$; 0.01\%), and BSD Cadets ($n=76{,}903$; 0.02\%).
Each object is a process node; atoms are event-type identifiers in the
process's execution trace.

\paragraph{General categorical benchmarks (7 datasets).}
To evaluate whether the proposed similarity representation generalizes
beyond cybersecurity, we include seven categorical anomaly-detection
workloads from ADRepository\footnote{https://www.dbs.ifi.lmu.de/research/outlier-evaluation/DAMI/}~\cite{Campos16,Pang16}.
\textsc{Reuters-Corn} represents documents through binary lexical
indicators; \textsc{W7A} is a sparse web-classification workload;
\textsc{Solar Flare} contains categorical solar-activity observations;
\textsc{Bank Marketing} contains nominalized direct-marketing records;
\textsc{APascal} uses image-derived semantic attributes;
\textsc{AID362} contains nominalized molecular-bioassay descriptors;
and \textsc{Internet Ads} is a high-dimensional web-advertisement
workload. Several of these datasets are imbalanced classification
problems converted into binary anomaly-detection tasks. Accordingly,
their minority classes are treated as anomalies for evaluation rather
than interpreted as naturally occurring anomalies in every domain.

\begin{table} 
\centering
\caption{Processed workload statistics.  AR denotes anomaly ratio.}
\label{tab:datasets}
\scriptsize
\begin{tabular}{llrr}
\toprule
Family & Workload & $n$ & AR (\%) \\
\midrule
\multirow{7}{*}{UWF-V1}
& Credential Access     & 91{,}400 & 47.59 \\
& Defense Evasion       & 48{,}232 & 0.68  \\
& Exfiltration          & 47{,}929 & 0.05  \\
& Initial Access        & 48{,}467 & 1.16  \\
& Persistence           & 48{,}232 & 0.68  \\
& Privilege Escalation  & 48{,}232 & 0.68  \\
& Reconnaissance        & 50{,}815 & 5.72  \\
\midrule
\multirow{7}{*}{UWF-V2}
& Credential Access & 16,762 & 3.3 \\
& Defense Evasion & 16,530 & 2.0 \\
& Exfiltration & 16,554 & 2.2 \\
& Initial Access & 16,762 & 3.3 \\
& Persistence & 16,530 & 2.0 \\
& Privilege Escalation & 16,530 & 2.0 \\
& Reconnaissance & 18,128 & 10.6 \\
\midrule
\multirow{2}{*}{NSL-KDD}
& Probe & 64,759 & 6.4 \\
& U2R & 60,821 & 0.4 \\
\midrule
\multirow{4}{*}{DARPA TC}
& Android & 102 & 8.8 \\
& Windows & 17,569 & 0.05 \\
& Linux & 272,376 & 0.01 \\
& BSD & 76,903 & 0.02 \\
\midrule
\multirow{4}{*}{Categorical}
& Reuters-Corn & 12898  &  1.83 \\
&  W7A & 49750  & 2.97  \\
&  Solar Flare & 1067  & 4.02  \\
&  Bank Marketing &  41189 & 11.26  \\
&  APascal &   12696&  1.38 \\
&  AID362 & 4280  & 1.38  \\
&  Internet Ads &  3280 & 14  \\
\bottomrule
\end{tabular}
\end{table}

\subsection{Protocol and configurations}
\label{sec:protocol}

\paragraph{Transductive evaluation protocol.}
The benchmark evaluation is \emph{transductive}: for each workload, the
rare-rule dictionary is constructed from the complete unlabeled collection
that is subsequently searched. No class labels are used during dictionary
construction, support estimation, rule mining, profile generation, or
parameter adaptation. The protocol therefore does not introduce label
leakage, although the learned representation reflects the empirical
distribution of the evaluated collection. In a prospective deployment,
the dictionary could instead be learned from a historical reference window
and then frozen, incrementally maintained, or periodically refreshed. We
discuss this limitation and its deployment implications in
Section~\ref{sec:discussion}.

\paragraph{Model configurations.}
A fixed instantiation of \method{} involves nine numerical parameters: the
upper rarity threshold $\tau$, the maximum minimal-rare-itemset length $L$,
the minimum rule-confidence threshold $c_{\min}$, the five weighting exponents
$\alpha,\beta,\gamma,\delta,\zeta$, and the smoothing constant $\eta$.
Stability estimation additionally uses $B=10$ subsamples. Rule weights
are computed according to Eq.~\eqref{eq:weight}. Unless explicitly varied, the
shared confidence and weighting parameters are fixed to
\begin{equation}
\begin{aligned}
c_{\min}&=0.95,\qquad
\alpha=\beta=\gamma=\zeta=1,\\
\delta&=0.5,\qquad
\eta=1,\qquad B=10.
\end{aligned}
\label{eq:weight_defaults}
\end{equation}

Unless explicitly labeled otherwise, \textbf{\method{}} denotes the adaptive
configuration defined in Section~\ref{sec:auto_config} and
Algorithm~\ref{alg:auto}. It first constructs a preliminary rule space using
\begin{equation}
\tau_0=0.01,\qquad L_0=4,\qquad c_{\min}=0.95,
\label{eq:preliminary_defaults}
\end{equation}
and then adapts the final values of $\tau$ and $L$ from the unlabeled
rarity-score distribution of the workload. The global coverage multiplier is fixed heuristically to $\kappa=5$ and
used unchanged for every reported workload. It was not selected using
class labels, validation queries, or a held-out validation set. The
workload-specific adaptation of $\tau$ and $L$ therefore depends only on
the unlabeled rarity-score distribution.

We retain \textbf{\methodfixed{}} as a fixed-parameter reference:
\begin{equation}
\tau=0.01,\qquad L=4,\qquad c_{\min}=0.95.
\label{eq:fixed_defaults}
\end{equation}
This reference configuration is used only in explicitly identified secondary
analyses. The main results do not use post-hoc oracle selection between the
adaptive and fixed configurations.

\paragraph{Query-conditioned retrieval protocol.}
For each workload, up to 200 labeled anomalies are sampled as queries using a
fixed random seed. The same query set and candidate collection are used for
every query-conditioned similarity method. Query labels are used only to define
the evaluation queries and relevance judgments; they are not used to mine the
RareSense dictionary, estimate rule weights, or adapt $\tau$ and $L$. The query
object itself is removed from its candidate set.

All methods are evaluated using identical binary relevance labels, under which
a candidate is relevant when it is labeled anomalous. Empty RareSense query
profiles are retained rather than excluded. Such queries contain no retained
rare evidence and therefore induce an all-zero evidence-based ranking under
the benchmark policy. Zero-score ties are resolved deterministically and
consistently across workloads. Retaining these queries avoids optimistic
nDCG@10 estimates that could arise from filtering out difficult or uncovered
queries.

\paragraph{Sensitivity and reproducibility.}
The sensitivity analyses in Section~\ref{sec:sensitivity} are conducted
on the UWF-V1 Privilege Escalation workload. The structural parameters
$\tau$ and $L$, together with the confidence and weighting parameters,
are varied one factor at a time while all remaining settings are held
fixed. These analyses characterize local parameter
behavior and are not used to select workload-specific test configurations.

RareSense and the atomic similarity baselines are deterministic once the data,
parameters, query set, and tie-breaking policy are fixed; they are therefore
reported without run-to-run standard deviations. Stochastic baselines use the fixed seeds and implementation settings specified
in the accompanying reproducibility material.
\subsection{Baselines}

\paragraph{Primary query-conditioned similarity baselines.}
The primary retrieval comparison uses four atom-level similarities evaluated
under exactly the same anomaly queries, candidate sets, binary relevance labels,
and nDCG@10 computation as \method{}:
\textbf{Jaccard}, \textbf{IDF-Jaccard},
\textbf{cosine}, and \textbf{TF--IDF cosine}.
Jaccard compares set overlap directly; cosine operates on the binary
incidence vectors. IDF-Jaccard assigns each atom an inverse-frequency weight
before computing weighted set overlap, while TF--IDF cosine computes cosine
similarity after the corresponding inverse-frequency reweighting. These are the
like-for-like baselines for the paper's central question: whether moving from
atomic coordinates to rare-rule coordinates improves query-conditioned anomaly
retrieval.

\paragraph{Secondary detector references.}
Classical anomaly detectors are AVF~\cite{koufakou2007avf},
FPOF~\cite{he2005fpof}, HBOS~\cite{goldstein2012hbos},
ECOD~\cite{li2022ecod}, COPOD~\cite{li2020copod}, and
OC-SVM~\cite{scholkopf2001svm}. Deep baselines are a feed-forward
AutoEncoder \cite{zhou2017anomaly}, DIF~\cite{xu2023dif}, and LUNAR~\cite{goodge2022lunar}.
These methods natively output a scalar anomaly score. They are evaluated as secondary references for the global
anomaly-ranking task using AUROC; they are not treated as
query-conditioned similarity baselines. This
is a useful operational reference---it asks whether ``show me the most
anomalous objects'' can substitute for query-conditioned search---but it is not
a like-for-like similarity comparison.

\paragraph{Dual evaluation protocol.}
\method{} supports both retrieval and global ranking from the same rare-rule
representation, but the two tasks require different comparisons. For retrieval,
the primary baselines are the four query-conditioned similarities above.
For global anomaly ranking, \method{} uses the scalar rare-evidence score
$A_R(x_i)$ and is compared with classical and deep anomaly detectors using
AUROC. Pairwise similarities do not natively define a query-independent scalar
anomaly score; constructing one would require an additional neighborhood-based
outlier model and would therefore constitute a different experimental task.
We consequently keep the retrieval and global-ranking comparisons separate.

\subsection{Metrics: nDCG first, AUROC second}

Our primary metric is normalized discounted cumulative gain at rank 10:
\begin{equation}
\mathrm{DCG}@k=\sum_{r=1}^{k}
\frac{2^{\mathrm{rel}_r}-1}{\log_2(r+1)},
\qquad
\mathrm{nDCG}@k=
\frac{\mathrm{DCG}@k}{\mathrm{IDCG}@k},
\label{eq:ndcg}
\end{equation}
where $\mathrm{rel}_r\in\{0,1\}$ in the primary binary-relevance evaluation.
nDCG rewards relevant anomalies more strongly when they occur near the top of
the list and is therefore aligned with limited analyst inspection budgets.
This is particularly important in Security Operations Centers (SOCs), where
analysts typically investigate only the highest-ranked alerts because examining
an entire ranking is costly and time-critical: by the time all events have been
reviewed, an attacker may already have completed data exfiltration or caused
substantial damage.

AUROC is reported as a secondary measure for the query-independent global
score in Eq.~\eqref{eq:anomaly_score}. AUROC measures pairwise
anomaly--normal separation over the entire ranking; it does not directly measure
whether relevant anomalies appear in the first few positions. We therefore do
not assign AUROC values to the native pairwise similarity baselines in the main
analysis: doing so would require adding a separate neighborhood anomaly-scoring
procedure and would conflate two different tasks.
\section{Results}
\label{sec:results}

\subsection{Primary result: top-ranked anomaly retrieval}

Tables~\ref{tab:main_retrieval} and \ref{tab:main_auroc} report the macro-average results over the 27
workloads. The central comparison is between \method{} and the four
query-conditioned atomic similarities, all evaluated with identical anomaly
queries and binary relevance.

\begin{table}[t]
\centering
\caption{Macro-average query-conditioned retrieval performance over
27 workloads.}
\label{tab:main_retrieval}
\small
\begin{tabular}{lc}
\toprule
Method & nDCG@10 $\uparrow$\\
\midrule
\method{} & \textbf{0.696}\\
TF--IDF cosine & \underline{0.645}\\
Cosine & 0.618\\
Jaccard & 0.603\\
IDF-Jaccard & 0.584\\
\bottomrule
\end{tabular}
\end{table}

\begin{table}[t]
\centering
\caption{Macro-average AUROC for global anomaly ranking over
27 workloads.}
\label{tab:main_auroc}
\small
\begin{tabular}{lc}
\toprule
Method & AUROC $\uparrow$\\
\midrule
\method{} & \textbf{0.850}\\
AutoEncoder & \underline{0.841}\\
HBOS & 0.838\\
AVF & 0.827\\
DIF & 0.819\\
\methodfixed{} & 0.812\\
FPOF & 0.769\\
OC-SVM & 0.735\\
ECOD & 0.722\\
COPOD & 0.717\\
LUNAR & 0.711\\
\bottomrule
\end{tabular}
\end{table}
Our primary interest is in whether truly anomalous objects are concentrated
near the top of the ranking, rather than merely achieving a good ordering over
the full candidate set. As illustrated in
Fig.~\ref{fig:ndcg_motivation}, this reflects many operational
anomaly-detection settings in which only a limited number of high-priority
cases can be inspected. For example, a Security Operations Center (SOC) may
investigate only the highest-ranked alerts before an attacker progresses further
in the kill chain; in medical screening, clinicians may prioritize the most
suspicious patients or physiological signals; and in large-scale IoT
monitoring, operators may need to inspect only a small subset of devices or
events among millions of observations. nDCG@10 is therefore well suited to this
objective because it rewards relevant anomalies more strongly when they appear
near the top of the ranking and progressively discounts lower-ranked results.

\begin{figure*}[h!]
    \centering
    \includegraphics[width=0.6\linewidth]{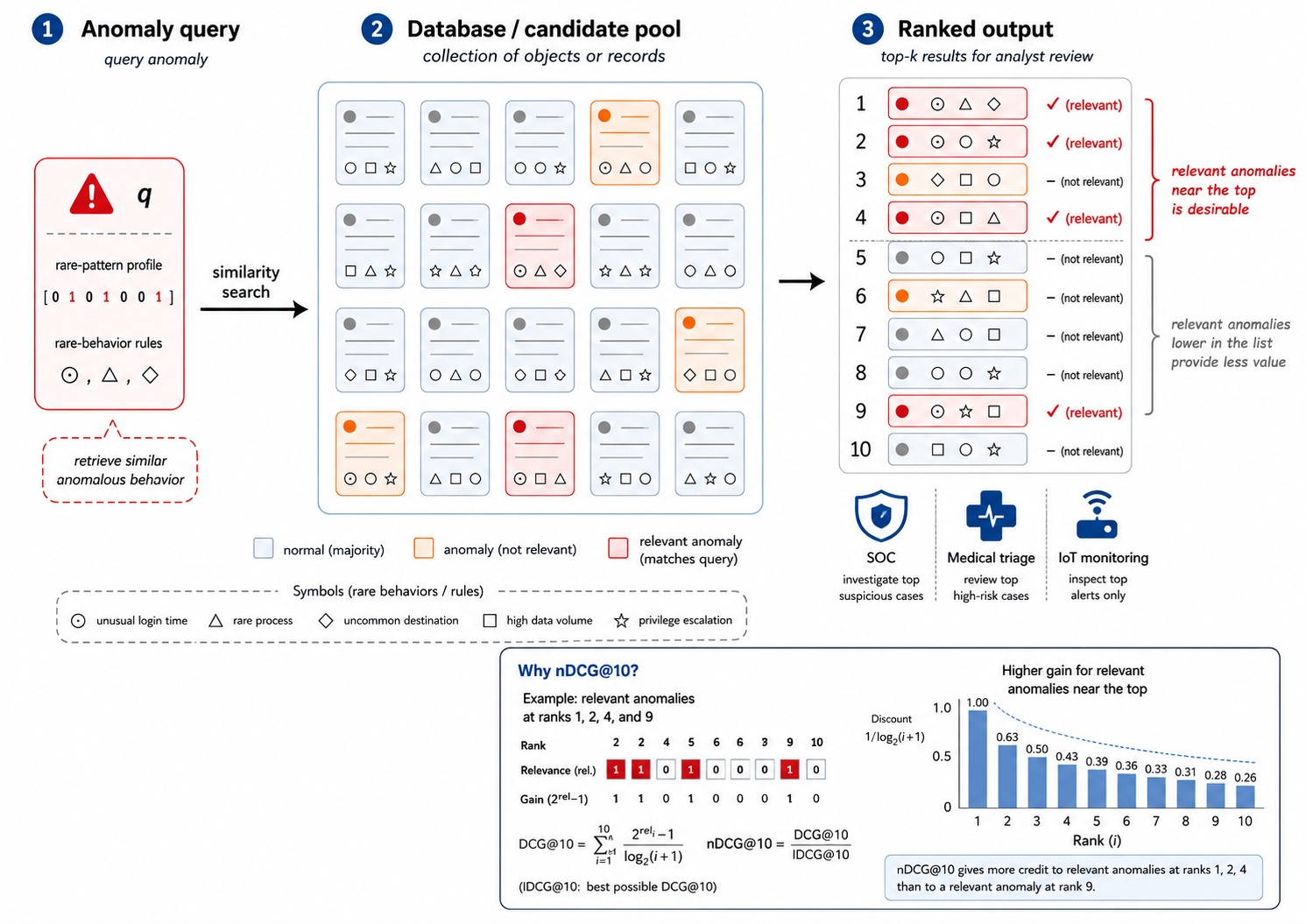}
\caption{Illustration of query-conditioned anomaly retrieval and the role of
top-ranked evaluation. An anomalous query is compared with a candidate
database using the learned similarity representation, producing a ranked list
of retrieved objects. Relevant anomalies appearing near the top of the ranking
are more valuable in operational settings such as Security Operations Centers
(SOCs), medical triage, and large-scale IoT monitoring, where only a limited
number of high-priority cases can typically be inspected. nDCG@10 captures this
requirement by assigning greater gain to relevant anomalies retrieved at higher
ranks and progressively discounting lower-ranked results.}
\label{fig:ndcg_motivation}
\end{figure*}

Across the 27 workloads, \method{} obtains the highest observed
macro-average nDCG@10, reaching approximately $0.696$ compared with
$0.645$ for TF--IDF cosine, the strongest atomic baseline on average.
This corresponds to an absolute improvement of approximately $0.051$
(5.1 percentage points), or about $7.9\%$ relative improvement.
\method{} also exceeds the strongest atomic similarity selected
independently for each workload on 18 of the 27 workloads.

For the secondary global anomaly-ranking task, \method{} obtains the
highest observed macro-average AUROC across the 27 workloads. However,
the corrected pairwise comparisons do not establish superiority over
several strong dedicated detectors. Thus, the main empirical contribution
remains top-ranked, query-conditioned retrieval rather than universal
dominance in global anomaly separation.

\subsection{Query-conditioned similarity by benchmark family}

Table~\ref{tab:similarity_family} separates the four benchmark families for
the like-for-like similarity comparison.

\begin{table*}[h]
\centering
\caption{Family-level nDCG@10 for the primary query-conditioned
similarity comparison. The Overall column is the macro-average over
all 27 workloads, whereas the Family-balanced column gives the four
benchmark families equal weight.}
\label{tab:similarity_family}
\small
\begin{tabular}{lrrrrrr}
\toprule
Method &
UWF (14) &
NSL-KDD (2) &
DARPA TC (4) &
General (7) &
Overall &
Family-balanced\\
\midrule
\method{}
& \textbf{0.904}
& \underline{0.826}
& \underline{0.262}
& \textbf{0.490}
& \textbf{0.696}
& \textbf{0.620}\\
Jaccard
& 0.792
& 0.672
& 0.179
& 0.449
& 0.603
& 0.523\\
IDF-Jaccard
& \underline{0.827}
& 0.397
& 0.228
& 0.353
& 0.584
& 0.451\\
Cosine
& 0.803
& \textbf{0.828}
& 0.167
& 0.444
& 0.618
& 0.561\\
TF--IDF cosine
& 0.808
& 0.820
& \textbf{0.278}
& \underline{0.476}
& \underline{0.645}
& \underline{0.596}\\
\bottomrule
\end{tabular}
\end{table*}

The family-level pattern is informative. On UWF, \method{} reaches
$0.904$, improving over the strongest atomic alternative,
IDF-Jaccard ($0.827$). This is the regime most aligned with the
method's hypothesis: anomalies share repeatable higher-order rare
conjunctions that are not fully captured by atom-wise overlap.

On NSL-KDD, \method{} ($0.826$) is essentially tied with cosine
($0.828$), indicating that the rare-rule transformation preserves
strong retrieval without providing a material advantage over the best
atomic geometry. On DARPA TC, TF--IDF cosine remains slightly stronger
($0.278$ versus $0.262$), consistent with the very small and
behaviorally heterogeneous attack populations in these provenance
workloads.

On the seven general categorical workloads, \method{} reaches a
macro-average nDCG@10 of approximately $0.490$, compared with $0.476$
for TF--IDF cosine, the strongest atomic baseline in this family.
\method{} achieves the best result on six of the seven workloads, with
AID362 providing the principal negative case. This extension indicates
that the benefit of rare-rule coordinates is not confined to
cybersecurity representations.

Because 14 of the 27 workloads come from UWF, we additionally compute a
family-balanced summary that gives UWF, NSL-KDD, DARPA TC, and the
general categorical family equal weight. \method{} remains strongest at
approximately $0.620$, compared with $0.596$ for TF--IDF cosine.
Therefore, the overall advantage is not solely attributable to the
larger number of UWF workloads.

\subsection{Per-workload results}

Tables~\ref{tab:q10_per_dataset} and~\ref{tab:auc_per_dataset} provide a
workload-level view of the two complementary capabilities evaluated in this
study. nDCG@10 measures whether relevant anomalies are concentrated near the
top of a \emph{query-conditioned} retrieval ranking, whereas AUROC evaluates
how well the query-independent score $A_R$ globally separates anomalous from
normal objects. The distinction is important: a method may rank anomalies
well globally without necessarily retrieving anomalies that share the same
rare evidence as a particular query, and conversely a highly informative
query may retrieve closely related anomalies even when the corresponding
global anomaly score is less discriminative.

\begin{table*}[t]
\centering
\caption{Per-workload nDCG@10 for the primary query-conditioned retrieval
comparison. All methods use the same anomaly queries, candidate sets, and
binary relevance labels.}
\label{tab:q10_per_dataset}
\scriptsize
\setlength{\tabcolsep}{4pt}

\begin{tabular}{llrrrrr}
\toprule
Family & Workload & RareSense & Jaccard & IDF-Jacc.  &
Cosine & TF--IDF Cos.\\
\midrule

\multirow{7}{*}{UWF-V1}
& Cred. Access   & \textbf{0.785} & 0.660 & \underline{0.770}  & 0.520 & 0.530\\
& Def. Evasion   & \textbf{0.803} & \underline{0.741} & 0.675  & \underline{0.741} & 0.668\\
& Exfiltration   & \textbf{0.833} & 0.623 & \underline{0.780}  & 0.651 & 0.771\\
& Initial Access & \textbf{0.976} & 0.904 & 0.893  & 0.904 & \underline{0.932}\\
& Persistence    & \textbf{0.978} & 0.963 & \underline{0.969}  & 0.963 & 0.910\\
& Priv. Esc.     & \textbf{0.973} & 0.674 & \underline{0.825}  & 0.774 & 0.771\\
& Recon.          & 0.799 & \textbf{0.832} & 0.802  & \textbf{0.832} & \underline{0.811}\\

\midrule

\multirow{7}{*}{UWF-V2}
& Cred. Access   & \textbf{0.978} & 0.361 & 0.441  & 0.401 & \underline{0.451}\\
& Def. Evasion   & \textbf{0.957} & 0.905 & \underline{0.906}  & 0.905 & \underline{0.906}\\
& Exfiltration   & 0.845 & \textbf{0.925} & 0.854  & \textbf{0.925} & \underline{0.896}\\
& Initial Access & \textbf{0.978} & \underline{0.921} & 0.920  & 0.919 & 0.912\\
& Persistence    & 0.957 & 0.961 & \textbf{0.971}  & 0.960 & \underline{0.966}\\
& Priv. Esc.     & 0.957 & 0.905 & \underline{0.958}  & 0.935 & \textbf{0.966}\\
& Recon.          & \textbf{0.840} & 0.713 & 0.819  & 0.818 & \underline{0.826}\\

\midrule

\multirow{2}{*}{NSL-KDD}
& Probe & \textbf{0.924} & 0.771 & 0.793  & \underline{0.884} & 0.875\\
& U2R   & 0.727 & 0.572 & 0.000  &
  \textbf{0.772} & \underline{0.765}\\

\midrule

\multirow{4}{*}{DARPA TC}
& Android & \textbf{0.665} & 0.154 & 0.247  & 0.062 & \underline{0.347}\\
& Windows & 0.128 & 0.271 & \underline{0.283} & 0.264 & \textbf{0.383}\\
& Linux   & 0.018 & 0.014  & \textbf{0.131} & \underline{0.066} &
  \textbf{0.131}\\
& BSD     & 0.235 & \textbf{0.277} & \underline{0.252}  & \textbf{0.277} & \underline{0.252}\\

\midrule

\multirow{7}{*}{General categorical}
& Reuters-Corn
& \textbf{0.640} & \underline{0.630} & 0.447 & 0.598 & 0.597\\
& W7A
& \textbf{0.764} & 0.620 & 0.431 & 0.623 & \underline{0.757}\\
& Solar Flare
& \textbf{0.199} & 0.165 & \underline{0.181} & 0.165 & 0.170\\
& Bank Marketing
& \textbf{0.261} & 0.156 & \underline{0.258} & 0.156 & 0.256\\
& APascal
& \textbf{0.671} & 0.654 & 0.411 & \underline{0.657} & 0.642\\
& AID362
& 0.113 & \textbf{0.137} & \underline{0.090} & \textbf{0.137}
& \textbf{0.137}\\
& Internet Ads
& \textbf{0.783} & \underline{0.781} & 0.650 & 0.769 & 0.775\\
\midrule
\multicolumn{2}{l}{\textbf{Macro-average}}
& \textbf{0.696} & 0.603 & 0.584 & 0.618 & \underline{0.645}\\
\midrule
\bottomrule
\end{tabular}
\end{table*}

Table~\ref{tab:q10_per_dataset} shows that the retrieval advantage of
\method{} is substantial on several workloads, but is not uniform.
Across the 27 workloads, RareSense achieves the highest
macro-average nDCG@10 of $0.696$, compared with $0.645$ for
TF--IDF cosine, the strongest atomic baseline on average.
It exceeds the strongest atomic similarity selected independently
for each workload on 18 of the 27 workloads.

The strongest improvements occur in regimes where rare co-occurrences appear
to provide more discriminative retrieval evidence than individual attributes.
For example, on UWF-V2 Credential Access, \method{} reaches
nDCG@10 $=0.978$, whereas the strongest atomic comparator reaches only
$0.451$. This means that relevant anomalies are placed very close to the top
of the returned list under the rare-rule representation, even though similarity
computed directly from individual attributes provides a much weaker ordering.
A similarly pronounced effect is observed on DARPA Android, where
\method{} obtains $0.665$ compared with $0.347$ for the best atomic baseline,
and on UWF-V1 Privilege Escalation, where the corresponding values are
$0.973$ and $0.825$.

Several workloads also approach near-ideal top-ranked retrieval.
For UWF-V1 Initial Access, Persistence, and Privilege Escalation,
nDCG@10 reaches $0.976$, $0.978$, and $0.973$, respectively, while
UWF-V2 Credential Access, Defense Evasion, and Initial Access obtain
$0.978$, $0.957$, and $0.978$. Values this close to one indicate that
relevant anomalies are concentrated very early in the ranking, which is the
operational regime targeted by \method{}: an analyst inspecting only the
first few returned objects is likely to encounter relevant anomalous examples
quickly.

The per-workload results also identify regimes in which rare-rule similarity
is less appropriate. On DARPA Windows and Linux, \method{} obtains
nDCG@10 values of only $0.128$ and $0.018$, compared with $0.383$ and
$0.131$ for TF--IDF cosine, respectively. Atomic similarities also remain
competitive on high-overlap workloads such as UWF-V2 Exfiltration,
Persistence, and Privilege Escalation. These negative cases are informative:
when relevant objects are characterized mainly by common individual attributes,
or when the rare-rule profiles become too sparse to create sufficient overlap,
the symbolic rarity representation may discard information that remains useful
to atomic similarity. The results therefore support a regime-dependent
interpretation of \method{} rather than a claim of universal dominance.
\begin{table*}[t]
\centering
\caption{Per-workload AUROC for the secondary global anomaly-ranking task.}
\label{tab:auc_per_dataset}
\scriptsize
\setlength{\tabcolsep}{3pt}

\resizebox{\textwidth}{!}{%
\begin{tabular}{llrrrrrrrrrrr}
\toprule
Family & Workload & RareSense & RS-Fixed & AVF & FPOF & HBOS &
ECOD & COPOD & OCSVM & AE & DIF & LUNAR\\
\midrule

\multirow{7}{*}{UWF-V1}
& Cred. Access
& \textbf{0.945} & 0.935 & 0.937 & 0.902 & \underline{0.938}
& 0.626 & 0.630 & 0.909 & 0.920 & 0.860 & 0.660\\

& Def. Evasion
& 0.957 & \textbf{0.986} & 0.947 & \underline{0.965} & 0.940
& 0.650 & 0.662 & 0.184 & 0.982 & 0.936 & 0.674\\

& Exfiltration
& \textbf{0.999} & 0.991 & 0.996 & 0.993 & 0.992
& 0.631 & 0.631 & 0.995 & \textbf{0.999} & \underline{0.997} & 0.644\\

& Initial Access
& 0.903 & \textbf{0.980} & 0.969 & 0.903 & 0.968
& 0.670 & 0.668 & 0.902 & 0.964 & \underline{0.971} & 0.666\\

& Persistence
& \textbf{0.986} & 0.972 & 0.941 & \underline{0.977} & 0.944
& 0.675 & 0.665 & 0.754 & 0.969 & 0.968 & 0.680\\

& Priv. Esc.
& \textbf{0.986} & 0.973 & 0.944 & \underline{0.978} & 0.945
& 0.665 & 0.665 & 0.665 & 0.945 & 0.968 & 0.670\\

& Recon.
& \textbf{0.929} & 0.749 & 0.906 & \underline{0.908} & 0.906
& 0.667 & 0.667 & 0.442 & 0.878 & 0.770 & 0.667\\

\midrule

\multirow{7}{*}{UWF-V2}
& Cred. Access
& 0.684 & 0.664 & 0.678 & \underline{0.792} & 0.697
& 0.640 & 0.641 & \textbf{0.962} & 0.622 & 0.200 & 0.669\\

& Def. Evasion
& 0.800 & 0.714 & 0.770 & \underline{0.910} & 0.740
& 0.660 & 0.668 & 0.714 & \textbf{0.954} & 0.830 & 0.670\\

& Exfiltration
& 0.992 & 0.986 & 0.996 & 0.993 & \textbf{0.999}
& 0.618 & 0.618 & 0.994 & \underline{0.998} & 0.997 & 0.716\\

& Initial Access
& 0.839 & 0.838 & 0.862 & 0.949 & 0.831
& 0.674 & 0.674 & \underline{0.952} & \textbf{0.956} & 0.900 & 0.676\\

& Persistence
& 0.800 & 0.714 & 0.780 & \underline{0.912} & 0.784
& 0.669 & 0.676 & 0.710 & \textbf{0.958} & 0.840 & 0.673\\

& Priv. Esc.
& 0.800 & 0.714 & 0.750 & \underline{0.911} & 0.754
& 0.675 & 0.660 & 0.718 & \textbf{0.950} & 0.854 & 0.673\\

& Recon.
& \textbf{0.824} & 0.740 & 0.633 & \underline{0.809} & 0.692
& 0.705 & 0.705 & 0.570 & 0.788 & 0.747 & 0.708\\

\midrule

\multirow{2}{*}{NSL-KDD}
& Probe
& 0.958 & 0.774 & 0.976 & \underline{0.971} & \textbf{0.977}
& \textbf{0.977} & 0.908 & 0.494 & 0.810 & 0.920 & 0.775\\

& U2R
& \textbf{0.984} & 0.936 & 0.883 & 0.875 & \underline{0.980}
& \underline{0.980} & 0.975 & 0.940 & 0.824 & 0.799 & 0.724\\

\midrule

\multirow{4}{*}{DARPA TC}
& Android
& \textbf{0.898} & 0.661 & \underline{0.879} & 0.401 & 0.867
& 0.594 & 0.594 & 0.308 & 0.815 & 0.780 & 0.786\\

& Windows
& 0.863 & \underline{0.985} & 0.969 & 0.685 & 0.984
& 0.984 & 0.984 & 0.994 & \textbf{0.997} & 0.990 & 0.946\\

& Linux
& 0.611 & 0.887 & 0.823 & 0.435 & 0.829
& 0.829 & 0.828 & \textbf{0.908} & 0.859 & \underline{0.900} & 0.847\\

& BSD
& 0.916 & 0.915 & 0.876 & 0.662 & 0.889
& 0.889 & 0.886 & 0.896 & \textbf{0.972} & \underline{0.967} & 0.966\\

\midrule
\midrule
\multirow{7}{*}{General categorical}
& Reuters-Corn
& \textbf{0.989} & 0.985 & 0.987 & 0.356 & 0.987
& \underline{0.988} & \underline{0.988} & 0.987 & 0.910 & 0.955 & 0.975\\

& W7A
& \textbf{0.721} & 0.405 & 0.473 & 0.332 & 0.531
& 0.568 & 0.564 & 0.498 & 0.510 & \underline{0.630} & 0.423\\

& Solar Flare
& 0.725 & 0.762 & 0.845 & \textbf{0.849} & 0.844
& 0.844 & \underline{0.846} & 0.797 & 0.718 & 0.796 & 0.745\\

& Bank Marketing
& \textbf{0.687} & 0.586 & 0.558 & 0.578 & 0.599
& 0.605 & 0.592 & 0.588 & 0.548 & 0.573 & 0.620\\

& APascal
& \textbf{0.719} & 0.658 & 0.620 & \underline{0.718} & 0.655
& 0.655 & 0.624 & 0.616 & 0.540 & 0.688 & 0.635\\

& AID362
& 0.643 & \textbf{0.667} & 0.630 & 0.472 & 0.648
& \underline{0.650} & 0.645 & \underline{0.650} & 0.601 & 0.601 & 0.582\\

& Internet Ads
& \textbf{0.780} & \underline{0.755} & 0.703 & 0.524 & 0.702
& 0.698 & 0.698 & 0.705 & 0.710 & 0.676 & 0.734\\
\midrule
\multicolumn{2}{l}{\textbf{Macro-average}}
& \textbf{0.850} & 0.812 & 0.827 & 0.769 & \underline{0.838}
& 0.722 & 0.717 & 0.735 & 0.841 & 0.819 & 0.711\\
\midrule
\bottomrule
\end{tabular}%
}
\end{table*}


\begin{figure*}[t]
\centering
\includegraphics[width=0.90\textwidth]
{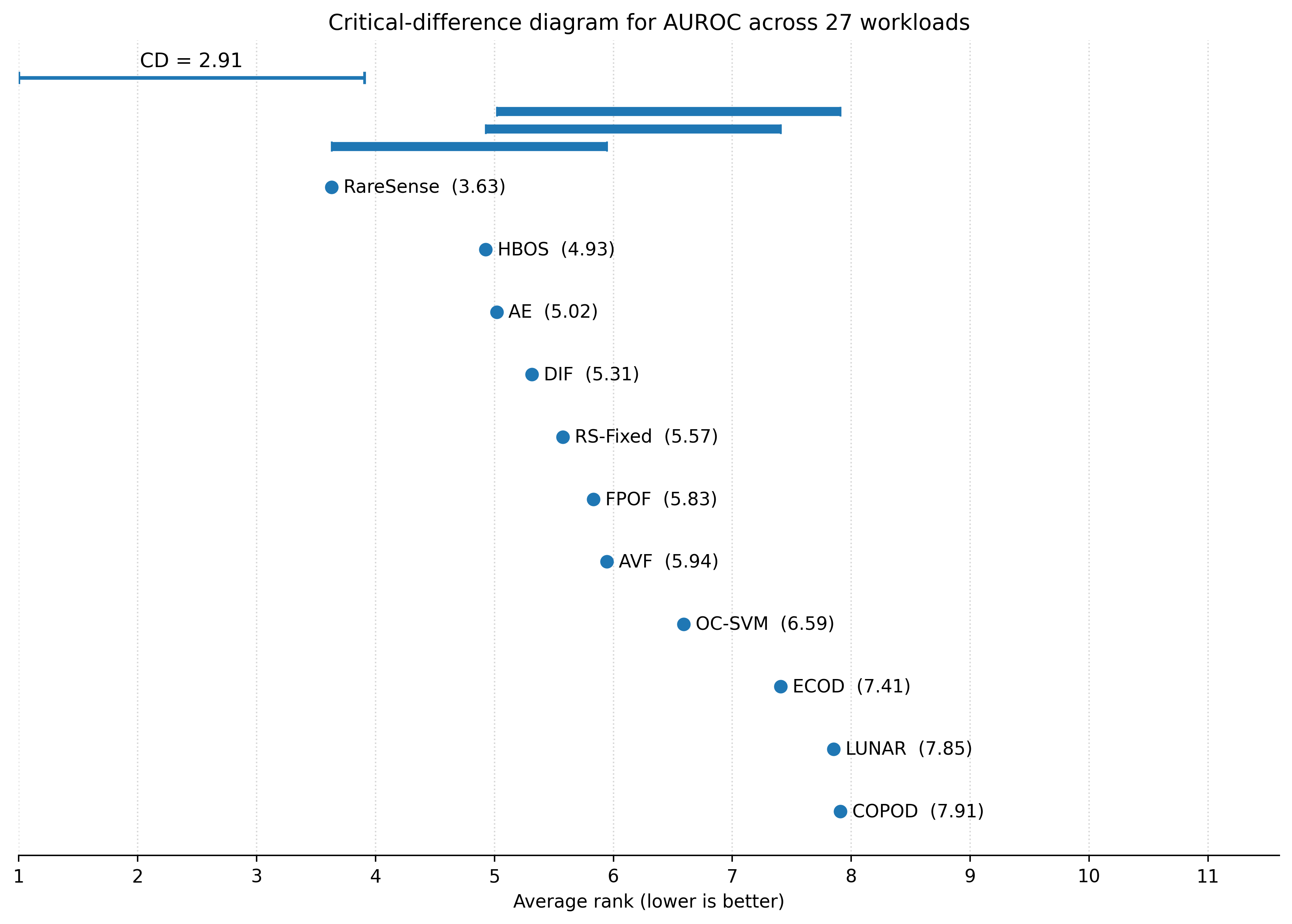}
\caption{Critical-difference diagram for global anomaly ranking based
on AUROC across the 27 evaluation workloads. Methods are ordered by
average rank, with lower ranks indicating better performance. The
Friedman test indicates significant overall differences among the
eleven methods
($\chi_F^2=44.00$, $p=3.29\times10^{-6}$).
Horizontal bars connect maximal groups of methods whose average-rank
differences do not exceed the Nemenyi critical difference
($\mathrm{CD}=2.91$, $\alpha=0.05$).}
\label{fig:auc_cd}
\end{figure*}

The global anomaly-ranking results in Table~\ref{tab:auc_per_dataset}
reveal a complementary aspect of the representation. Across the 27
workloads, \method{} achieves the highest observed macro-average AUROC
of $0.850$, followed by the AutoEncoder ($0.841$), HBOS ($0.838$),
AVF ($0.827$), and DIF ($0.819$). The rank-based analysis nevertheless
shows that \method{} is not statistically distinguishable from several
of these strong detectors. Thus, although the same rare-rule evidence
provides a competitive query-independent anomaly signal, the principal
contribution of \method{} remains query-conditioned retrieval.

Performance is particularly strong on several workloads. For UWF-V1
Exfiltration, \method{} reaches AUROC $=0.999$, indicating almost perfect
global separation, while Persistence and Privilege Escalation both reach
$0.986$. On NSL-KDD U2R, AUROC is $0.984$, and on DARPA Android it reaches
$0.898$. These results suggest that highly weighted rare-rule activations
often capture structures that are not only useful for retrieving related
anomalies but are also globally characteristic of anomalous behavior.

At the same time, the joint analysis of nDCG@10 and AUROC exposes an important
difference between \emph{retrieval quality} and \emph{global anomaly
discrimination}. DARPA Windows provides a clear example: \method{} obtains
AUROC $=0.863$ but nDCG@10 $=0.128$. The global score can therefore
distinguish anomalous objects reasonably well from normal ones, yet the rare
profiles of individual anomaly queries do not overlap sufficiently to retrieve
other anomalies near the top of the similarity ranking. DARPA Linux shows the
same phenomenon more strongly, with AUROC $=0.611$ but nDCG@10 only $0.018$.
This demonstrates that successful anomaly detection does not automatically
imply successful anomaly-to-anomaly retrieval.

The converse pattern is also visible. On UWF-V2 Credential Access,
\method{} obtains an exceptionally high nDCG@10 of $0.978$ while its AUROC
is only $0.684$. In this case, once an anomalous query is available, its
rare-rule profile provides an excellent basis for finding related anomalies,
even though the scalar score $A_R$ alone is less effective at globally
separating all anomalies from normal observations. This is precisely why the
two evaluation tasks should not be conflated: nDCG@10 evaluates the local,
query-conditioned organization of the similarity space, whereas AUROC
evaluates the global ordering induced by a single anomaly score.

Taken together, the two metrics suggest that the main strength of \method{}
lies in the \emph{structure} of the learned rare-rule space. Its strongest
advantage appears when anomalous objects share distinctive higher-order
evidence that allows them to be brought close together in a query-conditioned
ranking. The same evidence frequently yields competitive global anomaly
discrimination, but this is a secondary consequence rather than the sole
objective of the representation. The complementary nDCG@10 and AUROC results
therefore support the intended dual use of \method{}: top-$k$ retrieval for
investigating anomalies similar to a given query, and global scoring for
prioritizing suspicious objects when no query is yet available.

The rank-based analysis in Fig.~\ref{fig:auc_cd} complements the
macro-average AUROC results. The Friedman test reveals significant
overall differences among the eleven methods
($\chi_F^2=44.00$, $p=3.29\times10^{-6}$).
\method{} obtains the best average rank of $3.63$, followed by HBOS
($4.93$), the AutoEncoder ($5.02$), and DIF ($5.31$).
Thus, in addition to obtaining the highest observed macro-average
AUROC, \method{} exhibits the most favorable average rank across the
27 heterogeneous workloads.

The Nemenyi critical difference is approximately $2.91$ at
$\alpha=0.05$. Under this criterion, \method{} is not statistically
distinguishable from HBOS, the AutoEncoder, DIF, \methodfixed{}, FPOF,
or AVF. Its average-rank difference exceeds the critical difference
relative to OC-SVM, ECOD, LUNAR, and COPOD. The results therefore
support the conclusion that \method{} is competitive with the
strongest dedicated anomaly detectors and significantly stronger than
several lower-ranked alternatives, rather than establishing universal
superiority over every detector.
\subsection{Adaptive RareSense versus the fixed reference}

The final proposed method is the adaptive \method{} configuration of
Section~\ref{sec:auto_config}; \methodfixed{} is retained only to quantify the
value of workload-specific rarity-scale adaptation. On the original 20 cybersecurity workloads for which the
fixed-reference retrieval comparison was conducted,
RareSense-Fixed reaches a macro-average nDCG@10 of $0.498$,
whereas adaptive RareSense reaches $0.768$. The gain is especially large on UWF-V2, where a
strict fixed rarity threshold produces profiles that are too sparse for reliable
pairwise overlap. This result is consistent with the support-sensitivity study:
the rarity level that best separates anomalies globally need not provide enough
shared evidence for query-conditioned retrieval.

The adaptation does not select among alternative configurations using
test performance. The global coverage multiplier is fixed heuristically
to $\kappa=5$ and used unchanged across workloads; only the unlabeled
workload statistic $\hat\rho$ changes. Thus \methodfixed{} and
\method{} are not competing oracle variants but, respectively, a fixed
reference and the single reported adaptive procedure.

\subsection{Sensitivity analysis}
\label{sec:sensitivity}
We conduct a one-factor-at-a-time sensitivity analysis on the same
processed UWF-V1 Privilege Escalation workload used in the main
evaluation. It contains 48{,}232 objects, including 326 anomalies
($0.68\%$). The analysis compares the fixed reference configuration,
\methodfixed{}, with the workload-specific configuration selected by the
adaptive procedure. The fixed configuration uses
$\tau=0.01$, $L=4$, and $c_{\min}=0.95$, and obtains
AUROC $=0.973$ and nDCG@10 $=0.883$. The adaptive configuration selects
$\tau=0.034$ and $L=3$ through
Eqs.~\eqref{eq:tau_auto}--\eqref{eq:L_auto}, without using test labels,
and obtains AUROC $=0.986$ and nDCG@10 $=0.973$.

Figure~\ref{fig:structural_sensitivity} examines the two parameters that
directly determine which rare structures enter the representation. The
$\tau$ sweep fixes $L=3$, while the $L$ sweep fixes $\tau=0.034$.
All remaining parameters are held at
$c_{\min}=0.95$, $\alpha=\beta=\gamma=1$, $\delta=0.5$, and $\eta=1$.

\begin{figure*}[pt!]
    \centering
    \includegraphics[width=0.49\linewidth]
        {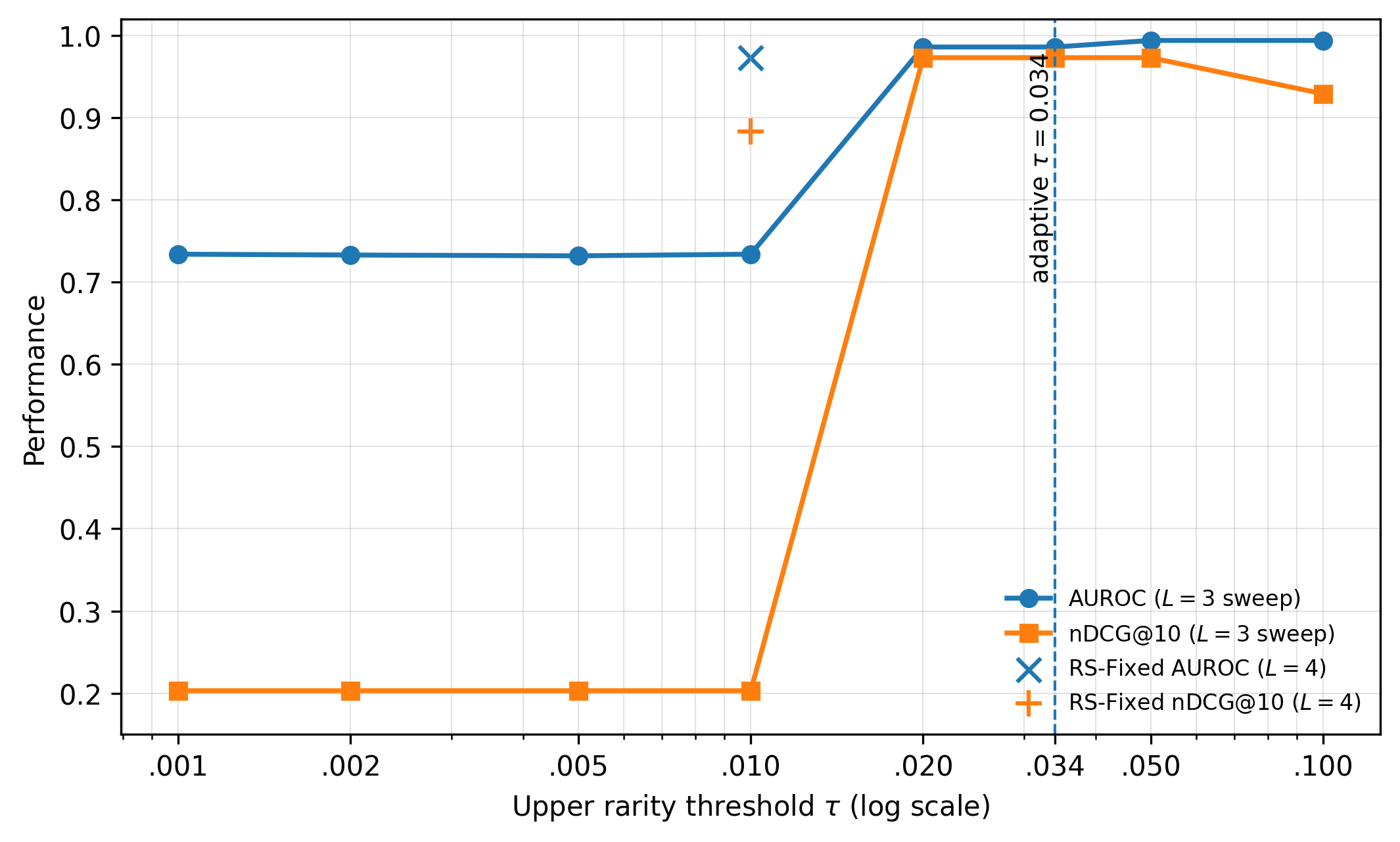}
    \hfill
    \includegraphics[width=0.49\linewidth]
        {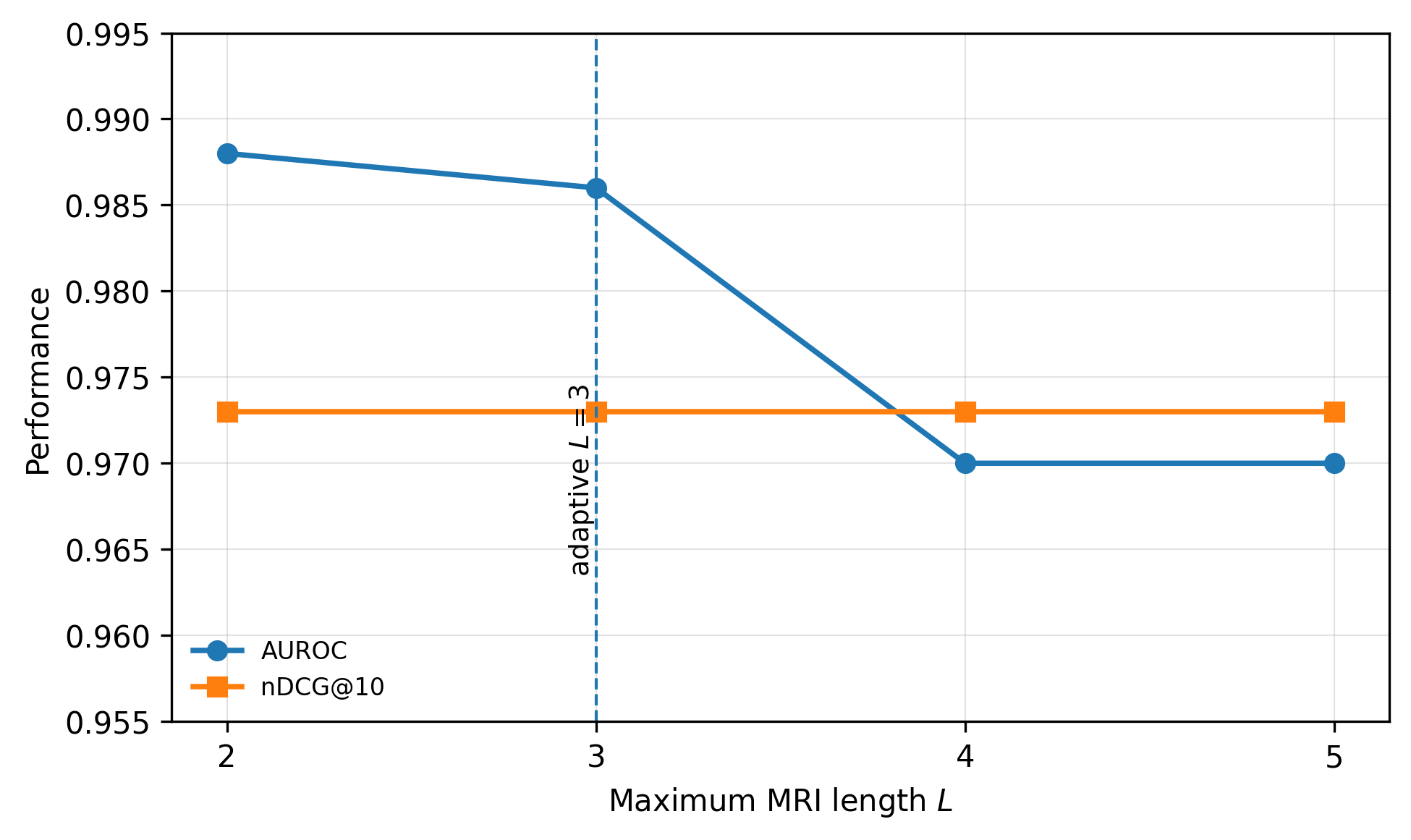}
    \caption{Sensitivity of \method{} to its structural parameters on
    UWF-V1 Privilege Escalation. \emph{Left:} AUROC and nDCG@10 as a
    function of the upper rarity threshold $\tau$, with $L=3$. The dashed
    line marks the adaptive value $\tau=0.034$. The separate RS-Fixed
    markers at $\tau=0.01$ use $L=4$ and are therefore not part of the
    one-factor $\tau$ sweep. \emph{Right:} sensitivity to the maximum MRI
    length $L$, with $\tau=0.034$. The dashed line marks the adaptive
    value $L=3$.}
    \label{fig:structural_sensitivity}
\end{figure*}

\paragraph{Rarity-threshold sensitivity.}
The support ceiling $\tau$ is the principal source of variation on this
workload. With $L=3$, restrictive thresholds
$\tau\leq0.01$ yield AUROC values near $0.73$ and nDCG@10 values near
$0.20$. Thus, although the retained evidence remains moderately informative
for global discrimination, the corresponding profiles provide insufficient
shared structure for effective anomaly-to-anomaly retrieval.

A sharp transition occurs at $\tau=0.02$, where AUROC increases from
$0.734$ to $0.986$ and nDCG@10 from $0.203$ to $0.973$. This indicates
that the useful operating regime begins only after the support ceiling admits
enough shared rare evidence to organize related anomalies into meaningful
neighborhoods. The adaptive value $\tau=0.034$ lies within this
high-performing regime and reproduces the principal \method{} result:
AUROC $=0.986$ and nDCG@10 $=0.973$.

Increasing the threshold to $\tau=0.05$ preserves nDCG@10 at $0.973$ and
raises AUROC slightly to $0.994$. At $\tau=0.10$, AUROC remains $0.994$,
but nDCG@10 decreases to $0.929$. Thus, the two objectives respond somewhat
differently once the threshold becomes broad: additional patterns continue to
support global anomaly separation, but begin to weaken the organization of
the top-ranked retrieval neighborhood. Overall, the interval
$0.02\leq\tau\leq0.05$ constitutes a favorable region for both tasks on
this workload.

The adaptive value is not presented as the post-hoc test optimum. In
particular, $\tau=0.05$ attains a slightly higher test AUROC. However,
selecting that value after inspecting the test labels would constitute
test-set tuning. In contrast, $\tau=0.034$ is obtained from the unlabeled
rarity-score distribution and falls inside the broad high-performing region
without using test labels.

\paragraph{Maximum-length sensitivity.}
The maximum MRI length $L$ is considerably less influential. Across
$L\in\{2,3,4,5\}$, nDCG@10 remains exactly $0.973$ at the reported
precision, while AUROC varies only between $0.970$ and $0.988$. The highest
observed AUROC occurs at $L=2$, with $0.988$, but the adaptive value
$L=3$ is very close at $0.986$. Increasing $L$ to four or five produces
only a small AUROC decrease and does not alter the top-$10$ retrieval
ordering.

These results suggest that interactions of length at most three are sufficient
to capture the relevant rare evidence on this workload. Allowing longer
itemsets does not improve retrieval, indicating that the useful neighborhood
structure is already represented by relatively short higher-order
co-occurrences.

\paragraph{Confidence and weighting parameters.}
Once the adaptive structural configuration
$(\tau,L)=(0.034,3)$ is fixed, performance is highly stable with respect
to the remaining parameters. Varying the confidence threshold
$c_{\min}$ from $0.50$ to $0.99$ leaves both AUROC and nDCG@10 unchanged
at $0.986$ and $0.973$, respectively. This indicates that the rules relevant
to the final ranking already satisfy comparatively high confidence levels.

The rarity exponent $\alpha$, confidence exponent $\beta$, and smoothing
constant $\eta$ similarly produce no visible change at the reported precision.
The lift exponent $\gamma$ changes AUROC only from $0.985$ to $0.988$, while
the length exponent $\delta$ changes it only from $0.986$ to $0.987$.
nDCG@10 remains $0.973$ throughout all of these sweeps.

The invariance of nDCG@10 does not imply that the numerical similarities are
identical. Rather, moderate parameter changes preserve the ordering of the
highest-ranked candidates, and nDCG@10 depends on that ordering rather than
on the absolute similarity values. The results therefore indicate that
performance is governed primarily by \emph{which rare structures are retained}
through $\tau$ and $L$, whereas moderate reweighting of already retained
coordinates has little effect on the resulting ranking.

Because this analysis varies one parameter at a time on a single workload,
it does not characterize interactions among parameters or establish universal
optimality across all datasets. It should therefore be interpreted as a local
diagnostic analysis. Nevertheless, it provides two clear observations:
the adaptive procedure selects a structural configuration within a broad
high-performing region, and the resulting similarity space is locally robust
to its confidence and weighting parameters.
\subsection{Ablation study}
\label{sec:ablation}

We ablate the main representational and weighting choices of \method{} on five
representative workloads. An important distinction concerns the role of
minimal rare itemsets (MRIs). MRIs are mined in all variants that construct
rule coordinates, because they provide the rare higher-order structures from
which candidate rules are generated. The ablation therefore does not ask
whether MRIs should be mined, but rather whether they should themselves be
retained as coordinates in the final similarity space.

In the \emph{rule-coordinate} representation used by the final method, an
object activates a coordinate associated with rule $r$ when its complete
evidence set satisfies $E(r)\subseteq T_i$. In the \emph{MRI-coordinate}
variant, each mined minimal rare itemset $p$ is instead used directly as a
binary coordinate, activated whenever $p\subseteq T_i$. The combined
representation retains both coordinate families. Thus, the comparison isolates
the effect of the final representation while preserving the same underlying
rare-pattern mining principle.

\begin{table}
\centering
\caption{Ablation on five representative workloads.}
\label{tab:ablation}
\tiny
\begin{tabular}{lcc}
\toprule
Configuration & nDCG@10 & AUROC \\
\midrule
Rule coordinates & \textbf{0.621} & \textbf{0.901}\\
Rule + MRI coordinates & 0.603 & 0.893 \\
No confidence filtering & 0.598 & 0.887 \\
No lift contribution & 0.587 & 0.881 \\
No length contribution ($\delta=0$) & 0.564 & 0.884 \\
No stability contribution ($\zeta=0$) & 0.550 & 0.801 \\
Uniform weights & 0.572 & 0.877 \\
No inverse-support factor & 0.561 & 0.869 \\
MRI coordinates & 0.441 & 0.812 \\
\midrule
IDF-Jaccard reference & 0.547 & 0.836 \\
\bottomrule
\end{tabular}
\end{table}

The representational ablation supports a separation between \emph{pattern
discovery} and \emph{similarity representation}. Using MRIs directly as
coordinates gives the weakest RareSense variant, with nDCG@10 $=0.441$ and
AUROC $=0.812$. Adding MRI coordinates to the rule representation also does
not improve performance: nDCG@10 decreases from $0.621$ to $0.603$ and
AUROC from $0.901$ to $0.893$, while introducing additional MRI-based coordinates into the final
representation. The best result is therefore obtained when MRIs serve only
as compact mining seeds and the final similarity space is formed from
reliability-qualified rule evidence.

This result should not be interpreted as evidence that rule
\emph{directionality} is directly exploited during object matching.
Activation depends on the complete evidence set $E(r)$ rather than on a
directional antecedent--consequent traversal. The advantage of rules arises
instead from the additional statistical criteria attached to them. Starting
from a rare co-occurrence, rule generation and filtering retain structures
supported by measures such as confidence and lift, thereby providing a
principled mechanism for selecting and weighting higher-order evidence before
it enters the similarity space.

The remaining ablations further support this interpretation. Removing
the confidence criterion reduces nDCG@10 from $0.621$ to $0.598$ and
AUROC from $0.901$ to $0.887$, while suppressing the lift contribution
reduces them to $0.587$ and $0.881$. Removing the length contribution
reduces nDCG@10 to $0.564$ and AUROC to $0.884$, showing that structural
complexity provides a modest but measurable contribution.

The strongest degradation occurs when stability weighting is removed:
nDCG@10 decreases to $0.550$ and AUROC to $0.801$. This indicates that
rediscovery across subsamples is important for suppressing fragile rules
that may arise from accidental low-support co-occurrences. Replacing all
learned weights by uniform weights yields nDCG@10 $=0.572$ and AUROC
$=0.877$, while removing the inverse-support factor gives $0.561$ and
$0.869$. Overall, the results support retaining all five components of
the composite weight, with stability and inverse support providing the
largest contributions.

Overall, the ablation supports the final design of \method{}: minimal rare
itemsets provide an efficient intermediate representation for discovering
unusual higher-order co-occurrences, while confidence-qualified and
weighted rule evidence provides a more selective final coordinate space.
The fact that the rule-coordinate representation also exceeds the
IDF-Jaccard reference ($0.621$ versus $0.547$ nDCG@10) further indicates
that its advantage is not explained solely by assigning greater weight to
individually rare attributes.
%

\subsection{Adversarial robustness}
\label{sec:adversarial}

We evaluate the robustness of \method{} under three black-box perturbation
families applied to anomalous objects. Benign objects and the learned rule
dictionary are held fixed throughout the experiment: the dictionary is mined
once from the original unperturbed data and is not re-estimated after an
attack. This setting models test-time behavioral modification by an adversary
who can alter its own observable behavior but does not have direct access to,
or control over, the learned rare-rule dictionary.

To isolate the effect of the perturbations from adaptive reconfiguration, this
diagnostic experiment uses the fixed reference configuration
\methodfixed{} with $\tau=0.01$, $L=4$, and $c_{\min}=0.95$.
For a metric $M$, we report
\[
\Delta M=M_{\mathrm{perturbed}}-M_{\mathrm{clean}},
\]
so negative values indicate degradation and positive values indicate improved
performance after perturbation. Results are averaged over four representative
workloads, comprising two UWF and two DARPA TC workloads.

\paragraph{Perturbation families.}
\textbf{Bit-flip noise ($\varepsilon$).}
Each binary atom of an anomalous transaction is independently flipped with
probability
$\varepsilon\in\{0.01,0.05,0.10,0.20\}$.
This provides a generic stress test for random telemetry corruption,
measurement noise, or unsophisticated tampering.

\textbf{Atom removal ($r$).}
For
$r\in\{1,2,3,5\}$,
$r$ active atoms are removed uniformly at random from each anomalous
transaction. This models partial behavioral suppression or missing telemetry
without assuming that the adversary knows which atoms participate in
high-weight RareSense rules. It is therefore a black-box feature-suppression
test rather than a dictionary-aware targeted attack.

\textbf{Benign mimicry ($c$).}
The
$c\in\{1,3,5,10\}$
most frequent atoms in the benign population are added to each anomalous
transaction. This models an adversary that pads its observable behavior with
common benign characteristics in an attempt to appear less exceptional.

\begin{table*}[t]
\centering
\caption{Adversarial robustness of \methodfixed{} under three black-box
perturbation families. Values are mean changes relative to the clean baseline
over four representative workloads (two UWF and two DARPA TC).
Negative values indicate degradation; values close to zero indicate limited
sensitivity to the tested perturbation.}
\label{tab:adversarial}
\small
\begin{tabular}{lrcc}
\toprule
Attack & Intensity & $\Delta$AUROC & $\Delta$nDCG@10 \\
\midrule
\multirow{4}{*}{Bit-flip noise ($\varepsilon$)}
 & 0.01 & $+$0.027 & $+$0.089 \\
 & 0.05 & $+$0.027 & $+$0.220 \\
 & 0.10 & $+$0.056 & $+$0.280 \\
 & 0.20 & $-$0.026 & $+$0.220 \\
\midrule
\multirow{4}{*}{Atom removal ($r$)}
 & 1 & $-$0.062 & \phantom{$+$}0.000 \\
 & 2 & $-$0.211 & $+$0.151 \\
 & 3 & $-$0.228 & $+$0.059 \\
 & 5 & $-$0.246 & $-$0.030 \\
\midrule
\multirow{4}{*}{Benign mimicry ($c$)}
 & 1  & \phantom{$+$}0.000 & \phantom{$+$}0.000 \\
 & 3  & \phantom{$+$}0.000 & \phantom{$+$}0.000 \\
 & 5  & $+$0.006 & \phantom{$+$}0.000 \\
 & 10 & $+$0.006 & $+$0.250 \\
\bottomrule
\end{tabular}
\end{table*}

\paragraph{Benign-feature injection has limited adverse effect.}
Mimicry produces essentially no degradation in global anomaly discrimination:
for $c\leq3$, $\Delta$AUROC is zero at the reported precision, and even for
$c\in\{5,10\}$ the mean change is only $+0.006$. Retrieval is similarly
unchanged for $c\leq5$.

This behavior is consistent with the frozen-dictionary construction. Injected
benign atoms cannot create new coordinates in the learned representation
because the rule dictionary is not re-mined after perturbation. They can only
change which \emph{existing} rules an object activates. Common benign atoms
therefore have limited influence unless, together with atoms already present
in the anomalous transaction, they complete the evidence set $E(r)$ of an
existing retained rule.

At $c=10$, nDCG@10 increases by $0.250$. This should not be interpreted as
evidence that mimicry necessarily improves detection. Because anomalous
objects participating in the retrieval evaluation are perturbed under the
same attack mechanism, shared perturbations can increase overlap between
anomalous query and candidate profiles. The result nevertheless shows that
the tested benign-padding strategy does not provide an effective evasion
mechanism against the frozen RareSense representation.

\paragraph{Random bit perturbations do not yield systematic evasion.}
For $\varepsilon\leq0.10$, both AUROC and nDCG@10 increase rather than
decrease. Random flips may remove some activated evidence, but they can also
introduce atoms that complete the evidence sets of existing rare rules.
Consequently, indiscriminate perturbation does not consistently move anomalous
objects toward a less suspicious representation.

Only at the largest perturbation level,
$\varepsilon=0.20$, does global discrimination deteriorate, with
$\Delta$AUROC $=-0.026$. Even then, nDCG@10 remains above the clean
baseline. The divergence between the two metrics again highlights that global
anomaly separation and query-conditioned neighborhood structure respond
differently to perturbations.

\paragraph{Evidence removal is the main global-ranking vulnerability.}
Atom removal produces the clearest degradation in AUROC. Removing one active
atom reduces AUROC by $0.062$, while removing two, three, and five atoms
produces mean changes of $-0.211$, $-0.228$, and $-0.246$, respectively.
This behavior follows naturally from the conjunction-based representation:
a single atom can participate in the evidence sets of several activated rules,
so removing it can simultaneously invalidate multiple coordinates.

The effect on retrieval is more moderate. nDCG@10 is unchanged for $r=1$,
increases for $r=2$ and $r=3$, and decreases only slightly
($-0.030$) for $r=5$. Thus, atom suppression is substantially more damaging
to the global anomaly score than to the top-ranked similarity structure under
the tested protocol. Because the removals are random rather than
dictionary-aware, these results do not represent a worst-case adaptive attack;
an adversary that knew which atoms support the highest-weight rules could
potentially construct stronger perturbations.

\paragraph{Summary.}
The perturbation study reveals an asymmetric robustness profile. Adding common
benign atoms has little adverse effect under a frozen rule dictionary, and
random bit perturbations do not provide a consistent evasion benefit. In
contrast, suppressing active evidence is the strongest tested vulnerability
for global anomaly ranking because removing a small number of atoms can break
multiple conjunction-based coordinates simultaneously. These experiments are
black-box robustness stress tests rather than worst-case adversarial guarantees.
Dictionary-aware attacks, clean-query versus perturbed-candidate evaluation,
and adaptive adversaries that optimize perturbations against rule activation
remain important directions for future work.

\subsection{Statistical significance}
\label{sec:significance}

We complement the aggregate performance results with non-parametric statistical
tests across the 27 evaluation workloads. Because the workloads differ
substantially in sample size, class imbalance, and feature space, we compare
methods using within-workload ranks rather than assuming normally distributed
performance differences.

We distinguish two comparisons. First, the primary retrieval analysis compares
\method{} with the four query-conditioned atomic similarities:
Jaccard, IDF-Jaccard, cosine, and TF--IDF cosine. Second, the broader
comparison evaluates the deployable RareSense configurations and scalar anomaly
detectors. The latter includes \method{}, \methodfixed{}, AVF, FPOF, HBOS,
ECOD, COPOD, OC-SVM, AE, DIF, and LUNAR. No post-hoc oracle configuration is
included in the statistical analysis.

\begin{table*}[t]
\centering
\caption{Friedman omnibus tests across the 27 evaluation workloads.}
\label{tab:significance}
\small
\begin{tabular}{lcccc}
\toprule
Comparison & Metric & Methods & $\chi_F^2$ & $p$\\
\midrule
Query-conditioned similarities
& nDCG@10 & 5 & 13.27 & 0.010\\
Global anomaly-ranking methods
& AUROC & 11 & 44.00 & $3.29\times10^{-6}$\\
\bottomrule
\end{tabular}
\end{table*}
\paragraph{Primary query-conditioned comparison.}
For the five query-conditioned methods evaluated across 27 workloads,
the Friedman test rejects the null hypothesis of equal performance
($\chi_F^2=13.27$, $p=0.010$). RareSense obtains the best average rank.
Paired Wilcoxon signed-rank comparisons between RareSense and each
atomic baseline, corrected using Holm's procedure, remain significant
at the $0.05$ level. These results support an overall statistical
advantage for RareSense while not implying that it dominates every
baseline on every individual workload.

\paragraph{Global anomaly-ranking comparison.}
For the eleven methods evaluated using AUROC across the 27 workloads,
the Friedman test rejects the null hypothesis of equal performance
($\chi_F^2=44.00$, $p=3.29\times10^{-6}$).
RareSense obtains the best average rank of $3.63$, followed by HBOS
($4.93$), the AutoEncoder ($5.02$), and DIF ($5.31$).
Under the Nemenyi criterion, RareSense is not statistically
distinguishable from HBOS, the AutoEncoder, DIF, RareSense-Fixed,
FPOF, or AVF. Its average-rank difference exceeds the critical
difference relative to OC-SVM, ECOD, LUNAR, and COPOD. These findings
support competitiveness with the strongest dedicated detectors rather
than universal superiority.

\subsection{Runtime analysis}
\label{sec:runtime}

We evaluate computational cost using wall-clock runtime on a representative
subset of workloads spanning different sample sizes and benchmark families.
All methods are executed on the same machine and software environment.
For adaptive \method{}, runtime includes the complete fitting procedure:
the preliminary mining pass used to estimate $\hat\rho$, adaptive parameter
selection, the final rule-space construction, profile generation, and scoring.

\begin{table*}[t]
\centering
\caption{Wall-clock runtimes in seconds on representative workloads.
Dataset sizes correspond to the canonical processed workloads in
Table~\ref{tab:datasets}. RareSense runtime includes the preliminary
adaptive pass, final dictionary construction, profile generation, and
scoring; query-time retrieval is measured over at most 200 anomaly
queries. 
}
\label{tab:runtime}
\small
\setlength{\tabcolsep}{4pt}
\begin{tabular}{lrrrrrrrrr}
\toprule
Dataset & RS & AVF & FPOF & HBOS & ECOD & COPOD & OCSVM & AE & DIF  \\
\midrule
UWF-DefEva-V1  & 1.4 & 0.3 & 12.1 & 0.1 & 1.2 & 0.8 & 142.3 & 31.2 & 189.4 \\
UWF-Recon-V1  & 2.1 & 0.4 & 14.3 & 0.2 & 1.5 & 0.9 & 163.2 & 38.7 & 221.3 \\
UWF-CredAcc-V1  & 5.3 & 0.8 & 31.2 & 0.3 & 2.8 & 1.9 & 910.4 & 87.4 & 512.1 \\
DARPA-5dir  & 0.1 & 0.1 & 0.8 & 0.0 & 0.2 & 0.1 & 6.2 & 6.1 & 38.4 \\
DARPA-Trace  & 4.8 & 1.1 & 45.2 & 0.5 & 6.3 & 4.1 & 920.1 & 91.3 & 3622.4 \\
DARPA-Cadets  & 1.8 & 0.5 & 18.4 & 0.2 & 2.1 & 1.3 & 243.7 & 42.1 & 483.2 \\
NSL-KDD-Probe   & 2.8 & 0.4 & 22.1 & 0.1 & 1.8 & 1.1 & 185.4 & 42.3 & 298.7 \\

\midrule

Mean & 2.6 & 0.5 & 20.6 & 0.2 & 2.3 & 1.4 & 367.3 & 48.4 & 766.5 \\
\bottomrule
\end{tabular}
\end{table*}

Runtime does not scale solely with the number of objects. Rare-itemset
and rule mining are output-sensitive: workloads of similar size can differ
in transaction density, support distributions, the number of candidate
itemsets examined, the number of generated closure and partition rules,
and the number of active object--rule incidences. 

Consequently, the
difference between UWF-DefEva-V1 and UWF-Recon-V1, despite their similar
numbers of objects, reflects differences in the mined search space rather
than an inconsistency in the timing protocol. Because each value is a
single wall-clock measurement, the table should be interpreted as an
indicative computational comparison rather than an estimate of runtime
variance.

\method{} requires a mean wall-clock time of 2.6\,s over the seven reported
workloads. It is slower than very lightweight frequency-based methods:
approximately $5\times$ slower than AVF and $13\times$ slower than HBOS on
average. This overhead reflects the additional cost of mining rare
higher-order structures and constructing the rule-profile representation.

Nevertheless, the absolute runtime remains modest and compares favorably with
more computationally intensive baselines. On average, \method{} is
approximately $8\times$ faster than FPOF, $141\times$ faster than OC-SVM,
$19\times$ faster than the autoencoder, and $295\times$ faster than DIF.
The gap becomes particularly pronounced on the largest workload shown:
on DARPA-Trace, \method{} completes in 4.8\,s compared with 920.1\,s for
OC-SVM and 3622.4\,s for DIF.

The computational cost of \method{} is concentrated in the offline
minimal-rare-itemset and rule-mining stages. Once the dictionary and sparse
profiles have been constructed, global scoring and query-time similarity
evaluation are comparatively inexpensive; in our implementation, post-mining
scoring requires less than 0.1\,s on the reported workloads. This separation
is operationally favorable for similarity-search applications because the
one-time symbolic mining cost can be amortized across many subsequent anomaly
queries.

Overall, the runtime analysis shows that \method{} occupies a useful middle
ground: it is more expensive than simple marginal-frequency detectors, but
substantially cheaper than several optimization- and learning-intensive
baselines while additionally providing an explicit query-conditioned
similarity space and rule-level explanations.

\subsection{Faithful Explanation of a Representative Retrieved Pair}
\label{sec:xai_case}

We illustrate the intrinsic explainability of \method{} using a
representative retrieval from the UWF Reconnaissance V2 workload.
The workload contains $n=18{,}128$ objects, of which 1{,}924
($10.6\%$) are labeled anomalous. The adaptive configuration selected
$\tau=0.10$, $L=5$, and $c_{\min}=0.95$, producing a dictionary of 1{,}296 retained rare-rule coordinates.

\paragraph{Case-selection protocol.}
To obtain an informative rather than degenerate explanation, we consider
anomaly queries whose top-ranked candidate satisfies four conditions:
(i) the query and candidate transactions are not identical,
(ii) their RareSense profiles are not identical,
(iii) their similarity lies strictly between zero and one, and
(iv) the top score is unique.
Among the eligible queries, we select the query whose per-query
nDCG@10 is closest to the median.

Ground-truth labels are used only to define the evaluation-query set and
assess retrieval relevance. Test-workload labels are not used in dictionary
construction, test-time parameter adaptation, similarity computation, or
candidate ranking. The calibration constant is fixed heuristically to $\kappa=5$ and is
used unchanged across all workloads without label-guided selection. Candidate
labels are inspected only after retrieval.

The selected query is \texttt{5}, with per-query nDCG@10 equal to
$0.63$. After excluding the query itself, \method{} retrieves
\texttt{13} at rank~1 with a unique score. Post-hoc inspection of the
evaluation ground truth identifies both objects as anomalous and associates
them with the T1595 Active Scanning category.

\paragraph{Transactions and rare-rule profiles.}
The query and retrieved candidate have different transactions and
different RareSense profiles. Table~\ref{tab:xai_transactions}
summarizes their shared and object-specific atoms.

\begin{table}[t]
\centering
\caption{Atomic overlap between the representative query $q$ and its
top-ranked candidate $x_1$. Numerical values denote the encoded bins or
categories produced by preprocessing, rather than literal byte, packet,
or port values.}
\label{tab:xai_transactions}
\small
\setlength{\tabcolsep}{4pt}
\begin{tabularx}{\columnwidth}{@{}lX@{}}
\toprule
Set & Atoms \\
\midrule
$T_q\cap T_{x_1}$
&
\texttt{conn\_state=2},
\texttt{dest\_ip=2},
\texttt{duration=3},
\texttt{local\_orig=2},
\texttt{local\_resp=2},
\texttt{orig\_bytes=2},
\texttt{orig\_pkts=4},
\texttt{proto=1},
\texttt{service=3},
\texttt{src\_ip=2}
\\[2pt]

$T_q\setminus T_{x_1}$
&
\texttt{resp\_bytes=3},
\texttt{resp\_pkts=4},
\texttt{port\_dst=7}
\\[2pt]

$T_{x_1}\setminus T_q$
&
\texttt{resp\_bytes=4},
\texttt{resp\_pkts=3}
\\
\bottomrule
\end{tabularx}
\end{table}

The query contains 13 atoms and activates six rare rules, whereas the
candidate contains 12 atoms and also activates six rare rules. Five rules
are shared:
\[
|P_q|=|P_{x_1}|=6,
\qquad
|P_q\cap P_{x_1}|=5.
\]
For every displayed shared rule $r$, activation was verified
programmatically:
\[
E(r)\subseteq T_q
\qquad\text{and}\qquad
E(r)\subseteq T_{x_1}.
\]

\paragraph{Exact score decomposition.}
The weighted shared evidence is
\[
W_{\cap}
=
\sum_{r\in P_q\cap P_{x_1}} w(r)
=
64.2.
\]
The query-only and candidate-only evidence have respective weights
\[
W_{q\setminus x}=15.4,
\qquad
W_{x\setminus q}=11.8.
\]
Consequently,
\begin{equation}
\begin{aligned}
S_R(q,x_1)
&=
\frac{W_{\cap}}
{W_{\cap}+W_{q\setminus x}+W_{x\setminus q}}\\
&=
\frac{64.2}{64.2+15.4+11.8}\\
&=
\frac{64.2}{91.4}
=
0.7024.
\end{aligned}
\label{eq:xai_score_decomposition}
\end{equation}

For each shared rule $r$, its exact contribution to the similarity is
\[
C_r(q,x_1)=\frac{w(r)}{W_{\cup}},
\qquad
W_{\cup}=91.4.
\]
These contributions satisfy
\[
\sum_{r\in P_q\cap P_{x_1}}C_r(q,x_1)
=
S_R(q,x_1).
\]
The unmatched fractions
\[
\frac{W_{q\setminus x}}{W_{\cup}}=0.1685,
\qquad
\frac{W_{x\setminus q}}{W_{\cup}}=0.1291
\]
are not positive contributions to the numerator. Instead, they represent
evidence found in only one profile and therefore enlarge the denominator
of the weighted Jaccard similarity.

Table~\ref{tab:xai_case_summary} reports the complete case metadata,
while Table~\ref{tab:xai_shared_rules} lists the five shared rules and
their exact score contributions.

\begin{table}[t]
\centering
\caption{Summary of the representative RareSense retrieval.
Test-workload labels are not used in dictionary construction,
test-time parameter adaptation, similarity computation, or candidate
ranking. Candidate labels are inspected only after retrieval.}
\label{tab:xai_case_summary}
\small
\setlength{\tabcolsep}{4pt}
\begin{tabular}{@{}lr@{}}
\toprule
Quantity & Value \\
\midrule
Workload & UWF Reconnaissance V2 \\
Objects & 18{,}128 \\
Anomalies & 1{,}924 ($10.6\%$) \\
Configuration &
$\tau=0.10$, $L=5$, $c_{\min}=0.95$ \\
Rule dictionary & 1{,}296 rules \\
\midrule
Query & \texttt{5} \\
Query nDCG@10 & $0.63$ \\
Query label & Anomalous; T1595 \\
$|T_q|$ & 13 atoms \\
$|P_q|$ & 6 rules \\
\midrule
Retrieved candidate & \texttt{13} \\
Candidate label & Anomalous; T1595 \\
$|T_{x_1}|$ & 12 atoms \\
$|P_{x_1}|$ & 6 rules \\
\midrule
$|T_q\cap T_{x_1}|$ & 10 atoms \\
$|P_q\cap P_{x_1}|$ & 5 rules \\
$W_{\cap}$ & 64.2 \\
$W_{q\setminus x}$ & 15.4 \\
$W_{x\setminus q}$ & 11.8 \\
$W_{\cup}$ & 91.4 \\
$S_R(q,x_1)$ & \textbf{0.7024} \\
\midrule
\method{} rank & 1; unique score \\
IDF-Jaccard rank & 15; unique score \\
IDF-Jaccard similarity & 0.311 \\
\bottomrule
\end{tabular}
\end{table}

\paragraph{Shared rare-rule explanation.}
Table~\ref{tab:xai_shared_rules} reports the five shared rules that form
the complete numerator of the RareSense similarity. Their observed
supports range from approximately $0.35\%$ to $0.52\%$ of the workload.
A confidence of 1.000 means that every observed occurrence of the
antecedent in the unlabeled workload was accompanied by the consequent.
It does not imply a universal deterministic relationship beyond the
observed data.

\begin{table*}[t]
\centering
\caption{Shared rare rules for the representative query--candidate pair.
\emph{Count} is the number of workload transactions containing the
complete evidence set. Contribution is
$C_r=w(r)/W_{\cup}$. The five contributions sum exactly to
$S_R(q,x_1)=0.7024$.}
\label{tab:xai_shared_rules}
\scriptsize
\setlength{\tabcolsep}{3pt}
\begin{tabularx}{\textwidth}{
@{}c>{\raggedright\arraybackslash}Xrrrrrr@{}}
\toprule
ID &
Antecedent $\Rightarrow$ Consequent &
Count & Supp. & Conf. & Lift & Weight & Contrib. \\
\midrule

$r_1$
&
\{\texttt{duration=3}, \texttt{orig\_bytes=2}\}
$\Rightarrow$
\{\texttt{dest\_ip=2}, \texttt{local\_orig=2},
\texttt{local\_resp=2}, \texttt{orig\_pkts=4},
\texttt{proto=1}, \texttt{service=3},
\texttt{src\_ip=2}\}
&
64 & 0.0035 & 1.000 & 4{,}127 & 24.8 & 0.2713
\\[2pt]

$r_2$
&
\{\texttt{conn\_state=2}, \texttt{orig\_pkts=4}\}
$\Rightarrow$
\{\texttt{dest\_ip=2}, \texttt{duration=3},
\texttt{local\_orig=2}, \texttt{local\_resp=2},
\texttt{service=3}, \texttt{src\_ip=2}\}
&
75 & 0.0041 & 1.000 & 3{,}592 & 17.3 & 0.1893
\\[2pt]

$r_3$
&
\{\texttt{duration=3}, \texttt{service=3}\}
$\Rightarrow$
\{\texttt{conn\_state=2}, \texttt{dest\_ip=2},
\texttt{local\_orig=2}, \texttt{src\_ip=2}\}
&
95 & 0.0052 & 1.000 & 2{,}874 & 11.9 & 0.1302
\\[2pt]

$r_4$
&
\{\texttt{conn\_state=2}, \texttt{duration=3},
\texttt{orig\_pkts=4}\}
$\Rightarrow$
\{\texttt{local\_orig=2}, \texttt{local\_resp=2},
\texttt{src\_ip=2}\}
&
75 & 0.0041 & 1.000 & 1{,}348 & 6.4 & 0.0700
\\[2pt]

$r_5$
&
\{\texttt{orig\_bytes=2}, \texttt{orig\_pkts=4},
\texttt{proto=1}\}
$\Rightarrow$
\{\texttt{duration=3}, \texttt{service=3}\}
&
95 & 0.0052 & 1.000 & 782 & 3.8 & 0.0416
\\

\midrule
\multicolumn{7}{r}{\textbf{Sum of shared-rule contributions}}
&
\textbf{0.7024}
\\
\bottomrule
\end{tabularx}
\end{table*}

\paragraph{Interpretation.}
The largest contribution comes from rule $r_1$, which accounts for
$0.2713$ of the final similarity. The rule captures the observed
co-occurrence of the encoded duration and origin-byte bins with a
seven-atom conjunction involving destination and source classes,
locality indicators, protocol, service category, and origin-packet bin.
Its complete evidence set occurs in 64 of the 18{,}128 workload objects.

Rules $r_2$--$r_5$ expose related but distinct higher-order
conjunctions involving connection state, duration, service category,
protocol, and encoded packet or byte bins. Their joint contribution
shows that the retrieval is not explained by one isolated attribute;
rather, it is driven by multiple overlapping rare conjunctions shared
by the two objects.

The query-only evidence is associated with the encoded combination
$\{\texttt{resp\_bytes=3},\texttt{resp\_pkts=4}\}$, whereas the
candidate-only evidence is associated with
$\{\texttt{resp\_bytes=4},\texttt{resp\_pkts=3}\}$.
These unmatched rules explain why the similarity is high but remains
below one.

\paragraph{Comparison with atomic similarity.}
IDF-Jaccard assigns the pair a similarity of 0.311 and ranks
$x_1$ at position~15, whereas \method{} assigns a similarity of
0.702 and ranks it first. IDF-Jaccard represents the pair through
weighted overlap between individual atoms. In contrast, \method{}
assigns weight to the shared higher-order rare conjunctions reported
in Table~\ref{tab:xai_shared_rules}, which contribute directly to the
pairwise score. The two methods therefore expose different forms of
evidence: atomic overlap for IDF-Jaccard and statistically qualified
rare conjunctions for \method{}.

\paragraph{Mechanistic faithfulness.}
We assess explanation faithfulness using a profile-level intervention.
The shared rules are removed from the query profile in decreasing order
of exact contribution, after which the similarity to $x_1$ and its
candidate rank are recomputed.

Removing $r_1$ reduces the similarity from 0.702 to 0.431 and moves
the candidate from rank~1 to rank~4. Removing the two largest
contributors reduces the similarity to 0.242 and moves the candidate
to rank~18. After removing the three largest contributors, the
similarity decreases to 0.112 and the candidate moves to rank~67.

For comparison, random deletion is evaluated exactly over all
$\binom{5}{k}$ subsets of $k$ shared rules. At $k=1$, ordered
deletion gives a similarity of 0.431, compared with a random-deletion
mean of $0.562\pm0.083$. At $k=2$, the corresponding values are
0.242 and $0.421\pm0.102$.

When all five shared rules are removed, the modified query and $x_1$
share no active rule, so their similarity becomes zero. The query-only
rule remains active in the modified query profile; therefore, other
candidates may still receive nonzero similarity. The resulting ordinal
position of $x_1$ is determined by the remaining candidate scores and
the deterministic tie-breaking policy and is therefore not interpreted.

The ordered intervention provides a mechanistic consistency check:
removing rules in decreasing exact-contribution order reduces both the
pairwise score and the candidate rank more rapidly than random removal.
This supports the faithfulness of the reported decomposition to the
RareSense retrieval mechanism.

\begin{table}[t]
\centering
\caption{Profile-level intervention under ordered and exhaustive random
deletion of shared rules. Random results report the exact mean and
population standard deviation over all $\binom{5}{k}$ subsets.
The candidate rank at $k=4$ is not reported because only the pairwise
similarity was retained for that intervention level.}
\label{tab:xai_deletion}
\small
\setlength{\tabcolsep}{4pt}
\begin{tabular}{@{}rrlr@{}}
\toprule
$k$ &
Ordered $S_R$ &
Candidate rank &
Random $S_R$ \\
\midrule
0 & 0.702 & 1 & $0.702\pm0.000$ \\
1 & 0.431 & 4 & $0.562\pm0.083$ \\
2 & 0.242 & 18 & $0.421\pm0.102$ \\
3 & 0.112 & 67 & $0.281\pm0.102$ \\
4 & 0.042 & not reported & $0.141\pm0.083$ \\
5 & 0.000 & not interpreted & $0.000\pm0.000$ \\
\bottomrule
\end{tabular}
\end{table}

\begin{figure*}[t]
\centering
\includegraphics[width=\textwidth]
{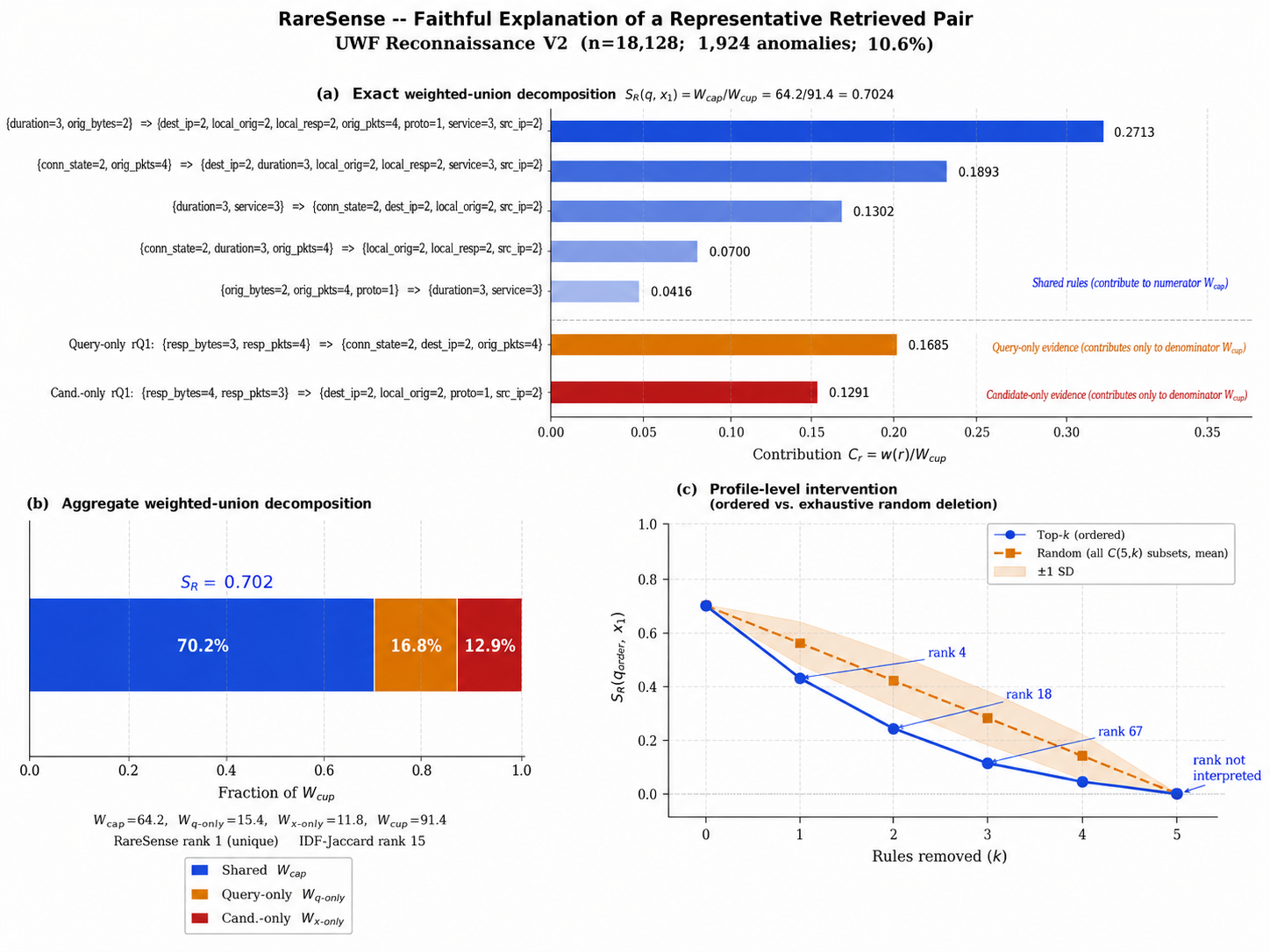}
\caption{\textbf{Faithful explanation of a representative RareSense
retrieval on UWF Reconnaissance V2.}
\textbf{(a)} Exact weighted-union decomposition. The five blue bars show
the normalized shared-rule contributions
$C_r=w(r)/W_{\cup}$ and sum to
$S_R(q,x_1)=0.7024$. The orange and red bars show the query-only and
candidate-only fractions of the weighted union. These unmatched
components enlarge the denominator but do not contribute to the
similarity numerator. All seven weighted-union fractions sum to one.
\textbf{(b)} Aggregate decomposition of the weighted union into shared
($70.2\%$), query-only ($16.8\%$), and candidate-only ($12.9\%$)
evidence.
\textbf{(c)} Profile-level intervention under ordered and exhaustive
random deletion of shared rules. Removing rules in decreasing exact
contribution order reduces the pairwise similarity and candidate rank
more rapidly than random deletion, supporting the mechanistic
faithfulness of the explanation. At $k=5$, the similarity to $x_1$
is zero; its resulting ordinal position is not interpreted.}
\label{fig:xai_case}
\end{figure*}

\section{Discussion}
\label{sec:discussion}

\paragraph{Why nDCG is the main story.}
RareSense is designed for search under a limited inspection budget. AUROC
averages pairwise anomaly--normal ordering over the entire ranking and can be
high even when the first relevant result is operationally too deep. nDCG@10
directly rewards concentrating relevant anomalies at the top. Under the
like-for-like query-conditioned protocol, \method{} reaches 0.696 versus 0.645 for TF--IDF cosine, the strongest atomic baseline on average. Its AUROC
advantage is much less pronounced, confirming that the contribution is
primarily a retrieval geometry rather than a universally superior detector.

\paragraph{Fair comparison across task types.}
The retrieval and detection tables intentionally use different primary
competitors. Atomic similarities are the correct baselines for the central
query-conditioned search task because their ranking changes with the query.
Scalar anomaly detectors are the correct baselines for the secondary global
ranking task because they natively produce query-independent anomaly scores.
Scalar anomaly detectors are evaluated only in the secondary global
ranking task and are therefore not interpreted as
query-conditioned similarity baselines.

\paragraph{What the method actually captures.}
RareSense does not simply assign larger weights to rare atomic features.  It
changes the coordinate system from atoms to reliability-qualified rare
conjunctions.  This makes it possible to distinguish objects that look equally
similar at the atomic level but differ in higher-order co-occurrence evidence
(Proposition~\ref{prop:tiebreak}).

\paragraph{What the method does not necessarily capture.}
RareSense retrieves objects that share rare-rule evidence, which is not
necessarily equivalent to sharing a broad semantic label such as an attack
tactic. In settings where relevance is primarily encoded by common atomic
features, suppressing those features may be counterproductive. Evaluating
same-tactic or same-attack-family retrieval remains an important direction
for future work.

\paragraph{DARPA failure mode.}
RareSense is less effective on several DARPA TC workloads. At the family
level, TF--IDF cosine slightly exceeds \method{} in nDCG@10
($0.278$ versus $0.262$), although \method{} performs substantially better
on Android ($0.665$ versus $0.347$). The weakest RareSense results occur
on Windows and Linux, where nDCG@10 reaches only $0.128$ and $0.018$,
respectively.

A plausible explanation is that the attack populations in these provenance
workloads are extremely small and behaviorally heterogeneous. When anomaly
processes activate different rare-rule subsets, they share insufficient
evidence for reliable query-conditioned retrieval. However, the present
experiments do not directly isolate profile sparsity or heterogeneity as the
causal source of the observed performance. A dedicated profile-coverage and
shared-rule analysis is therefore left for future work. These results motivate
hybrid atomic--rare representations and multi-query retrieval rather than a
claim of universal rare-rule superiority.

\paragraph{Heuristic adaptive calibration without label selection.}
The adaptive rarity scale is driven by an unlabeled upper-tail statistic
of a preliminary RareSense fit. The global coverage multiplier is fixed
heuristically to $\kappa=5$ and used unchanged for all workloads; it was
not selected using class labels, validation queries, or held-out
validation data. Consequently, only the unlabeled statistic $\hat\rho$
determines the workload-specific values of $\tau$ and $L$. The preliminary
fit adds an explicit extra mining pass, which is included in the reported
runtime.

\paragraph{Transductive evaluation.}
The reported dictionaries are mined unsupervised from the evaluation
collection.  No labels are used, but the setup is transductive.  In deployment,
the dictionary should be learned from historical reference data and then frozen
or updated on a schedule.  A future inductive evaluation should explicitly
separate mining and retrieval periods.

\paragraph{High-contamination workloads.}
One UWF workload has an anomaly ratio close to 48\%.  Such a setting violates
the everyday intuition that anomalies are globally rare.  We retain it as a
stress test because RareSense mines rare \emph{combinations}, not rare class
labels, but results on this workload should not be interpreted as a realistic
contamination scenario.

\paragraph{Indexing and scale.}
The present contribution is the similarity model, not a new index structure.
Metric and inverted-index compatibility provide a path to scalable exact or
approximate retrieval, but the current experiments do not claim a new
state-of-the-art indexing algorithm.

\paragraph{Explainability.}
The explanation is faithful by construction: shared-rule contributions sum
exactly to the similarity.  This is stronger than a post-hoc feature attribution
for the specific question ``why was this neighbor retrieved?''  It does not,
however, prove causal meaning of the rules; domain interpretation remains the
responsibility of the analyst.

\section{Reproducibility}
\label{sec:repro}

Unless otherwise stated, the reported method is adaptive \method{}.
Its preliminary fit uses $\tau_0=0.01$, $c_{\min}=0.95$, and $L_0=4$;
$\hat\rho$ is computed by Eq.~\eqref{eq:rho_hat}, the final support ceiling
by Eq.~\eqref{eq:tau_auto}, and $L$ by Eq.~\eqref{eq:L_auto}. The adaptive support ceiling uses the fixed heuristic multiplier
$\kappa=5$, corresponding in the implementation to
$\tau=\operatorname{clip}(5\hat{\rho},0.001,0.10)$. The multiplier was
not selected using class labels, validation queries, or a held-out
validation set, and the same value is used for every workload. The fixed reference \methodfixed{} uses $\tau=0.01$, $L=4$, and
$c_{\min}=0.95$. Rule weights use $\eta=1$,
$\alpha=\beta=\gamma=\zeta=1$, and $\delta=0.5$ in
Eq.~\eqref{eq:weight}. Stability is estimated from $B=10$
subsamples; note
that the weight exponent $\alpha$ is distinct from the adaptive coverage
multiplier $\kappa$.

Both partition rules and exact closure rules are generated in the reported
implementation. Query sets are fixed and reused across all retrieval methods;
Jaccard, IDF-Jaccard, cosine, TF--IDF cosine, and \method{} therefore see
identical queries, candidate sets, and binary relevance labels in the primary
comparison. Empty RareSense query profiles are retained and evaluated under the
fixed all-zero-similarity tie policy rather than discarded.

RareSense and the atomic similarity baselines are deterministic once the data,
parameters, query set, and tie-breaking policy are fixed. A common fixed random
seed of 42 is used whenever stochastic initialization or subsampling applies,
including for the AutoEncoder, DIF, LUNAR, and the subsampled OC-SVM
configuration. The standard deviations reported in
Table~\ref{tab:xai_deletion} are exact population standard deviations over
the corresponding deletion subsets, not measures of run-to-run variability.

Labels and attack-category fields are excluded from transaction construction
and from all test-time parameter adaptation.

The preprocessing scripts, experiment configurations, processed
transaction matrices, canonical result files, and scripts used to
regenerate the reported tables, figures, macro-averages, and statistical
tests are available from the corresponding author upon reasonable request.
\section{Conclusion}
\label{sec:conclusion}

We introduced \method{}, a rarity-aware similarity model for sparse
transactional anomaly search.  Minimal rare itemsets serve as an intermediate
mining substrate; reliable association-rule evidence defines the final
coordinates.  Objects are compared through weighted Jaccard overlap in the
rare-rule space, and the same shared rules provide an exact additive explanation
of every retrieved neighbor.

The empirical conclusion is deliberately nDCG-first. Across 27
workloads from four benchmark families, the label-free adaptive
\method{} configuration reaches a macro-average nDCG@10 of
approximately $0.696$, compared with $0.645$ for TF--IDF cosine, the
strongest atomic baseline on average. \method{} exceeds the strongest
atomic similarity on 18 of the 27 workloads. The omnibus comparison
is statistically significant, and corrected paired comparisons favor
\method{} over each atomic baseline.

The gain is strongest on UWF, remains positive across the general
categorical benchmarks, is effectively tied with the strongest atomic
similarity on NSL-KDD, and is weaker on several DARPA workloads. For
the secondary global-ranking task, \method{} obtains the highest
observed macro-average AUROC while remaining statistically comparable
to several strong dedicated anomaly detectors.

The results therefore support a precise rather than universal claim:
rare-rule coordinates provide a useful complementary similarity geometry when
related anomalies share repeatable rare combinatorial evidence. When relevance
is encoded by common atomic features, or when rare-rule overlap is too sparse,
atomic similarity can be preferable.

Future work will combine atomic and rare-rule coordinates in a dual-granularity
space, evaluate inductive historical-to-future retrieval, and implement
specialized inverted/metric indexes for sub-second search at larger scale.

\section*{Statements and Declarations}
\subsection*{Competing Interests}
The authors have no relevant financial or non-financial interests to
disclose.

\subsection*{Data and Code Availability}
All datasets used in this study are publicly available. The DARPA
Transparent Computing/ADAPT data are available at
\url{https://gitlab.com/adaptdata}. The UWF datasets are available at
\url{https://datasets.uwf.edu/}. The NSL-KDD dataset is available at
\url{https://www.kaggle.com/datasets/hassan06/nslkdd}. The general
categorical benchmarks are available through ADRepository at
\url{https://www.dbs.ifi.lmu.de/research/outlier-evaluation/DAMI/}.
The code and reproducibility materials used to produce the reported
results are available from the corresponding author upon reasonable
request.

\subsection*{Author Contributions}
SB, TR contributed to conceptualization,
methodology, software, formal analysis, investigation, visualization,
and preparation of the original manuscript. Both authors read and approved the final
manuscript.

\subsection*{Funding}
No funds, grants, or other support were received for conducting this
study. 
\subsection*{Ethics Approval}
Not applicable. This study uses publicly available benchmark datasets
and does not involve human participants or animals.

%
%

%

\bibliographystyle{ieeetr}
\bibliography{bibliography}

\end{document}